\RequirePackage{fix-cm}
\documentclass[11pt,runningheads]{svjour3}

\usepackage{amssymb,amsmath,amscd}
\usepackage{bbm}
\usepackage{a4wide}
\usepackage{graphicx}  
\usepackage{epsfig,psfrag}  

\newtheorem{prop}{Proposition}

\newtheorem{coro}{Corollary}

 \usepackage[numbers,sort&compress]{natbib}
\renewcommand{\qed}{\hspace*{\fill}$\square$}
\newcommand{\ts}{\hspace{0.5pt}}

\newcommand{\one}{\mathbbm{1}}

\newcommand{\btimes}{\mbox{\huge \raisebox{-0.2ex}{$\times$}}}

\DeclareMathOperator{\sep}{sep}

\begin{document}

\title{ Single--crossover recombination 
in discrete time}
\author{Ute von Wangenheim \and Ellen Baake \and Michael Baake}



\institute{Ute von Wangenheim \at Technische Fakult\"at, Universit\"at Bielefeld, 
Box 100131, 33501 Bielefeld, Germany \\ \email{uvonwang@techfak.uni-bielefeld.de}
\and Ellen Baake \at Technische Fakult\"at, Universit\"at Bielefeld, 
Box 100131, 33501 Bielefeld, Germany \\ \email{ebaake@techfak.uni-bielefeld.de}
\and Michael Baake \at Fakult\"at f\"ur Mathematik, Universit\"at Bielefeld, 
Box 100131, 33501 Bielefeld, Germany \\ \email{mbaake@math.uni-bielefeld.de} }

\maketitle

\begin{abstract}
Modelling the process of recombination leads to a large coupled nonlinear
dynamical system. Here, we consider a particular case of recombination
in {\em discrete} time,
allowing only for {\em single crossovers}.
While the analogous dynamics in {\em continuous} time
admits a closed solution \citep{reco}, this no longer works 
for discrete time.
A more general model (i.e. without the restriction
to single crossovers) has been studied before
\citep{Bennett, Kevin1, Kevin2}
and was solved  algorithmically by means of
Haldane linearisation.
Using the special formalism introduced in \cite{reco},
we obtain further insight into the 
single-crossover dynamics
and the particular difficulties that arise in discrete time.
We then transform the equations
to a solvable system in a two-step procedure: linearisation followed by
diagonalisation. 
Still, the coefficients of the second step must
be determined in a recursive manner, but once this
is done for a given system, they allow for
an explicit solution valid for all times.
 \end{abstract} 



\section{Introduction}
The dynamics of the genetic composition of populations evolving
under recombination has been a long-standing subject of research.
The traditional models assume random mating, 
non-overlapping generations (meaning discrete time), and populations
so large that stochastic fluctuations may be neglected
and a law of large numbers (or infinite-population limit) applies.
Even this highly idealised setting leads to models that are
notoriously difficult to treat and solve, namely, to large
systems of coupled, nonlinear  difference equations. Here, the nonlinearity
is due to the random mating of the partner individuals involved in
sexual reproduction. 

Elucidating the underlying structure and finding solutions to
these equations has been a challenge
to theoretical population geneticists for nearly a century now.
The first studies go back to Jennings~\cite{Jennings} 
in 1917 and Robbins~\cite{Robbins} in 1918.
Building on \cite{Jennings}, Robbins solved the dynamics for two 
diallelic loci (to be called sites
from now on) and gave an explicit
formula for the (haplo)type frequencies as functions of time.
Geiringer~\cite{Geiringer}  investigated the general recombination model for an
arbitrary number of loci and for arbitrary
`recombination distributions' (meaning collections of probabilities
for the various partitionings of the sites that may occur during
recombination) in 1944.
She was the first to state the general form of the solution
of the recombination equation (as a convex combination of all
possible products of  certain marginal frequencies derived from
 the initial population)
and developed a method for
the recursive evaluation of the corresponding coefficients.
This simplifies the calculation of the type frequencies
at any time compared to the direct evaluation through
successive iteration of the dynamical system.
Even though she worked out the method for the
general case in principle, its evaluation becomes 
quite involved for more than three sites.

Her work was followed by  Bennett~\cite{Bennett} in 1954. He introduced a 
multilinear transformation of the type frequencies to certain functions that he
named  \emph{principal components}.
They correspond to linear combinations of certain correlation functions
that transform the dynamical system (exactly) into
a linear one.  The new variables decay independently and geometrically
for all times,
whence they decouple and diagonalise the dynamics.  They therefore 
provide an elegant 
solution in principle,  but the price to be paid
is that  the coefficients 
of the transformation must be constructed via recursions 
that involve the parameters of the recombination model.
Bennett worked this method out for up to six sites, but  did not give an 
explicit method for an arbitrary number of sites.
The approach was later completed within the  systematic framework of
genetic algebras, where it became known as \emph{Haldane linearisation},
compare \citep{Lyubich, HaleRingwood}.
But, in fact, Bennett's program may be completed outside this abstract
framework, as was shown by
Dawson~\citep{Kevin1, Kevin2}, who derived a general and explicit 
recursion for the coefficients
of the principal components. 
However, the proofs are somewhat technical and do not reveal
the underlying mathematical structure.
It is the aim of this paper to provide a more systematic,
but still elementary, approach 
that exploits the inherent (multi)linear and combinatorial structure 
of the problem --- at least for one particular, but biologically relevant,
special case, which will now be described.
Our special case  is
obtained by the restriction to single crossovers, which leads to what
we call \emph{single-crossover recombination} (SCR).
This is the extreme case of the biological phenomenon of 
\emph{interference}, and describes the situation where  
a crossover event completely inhibits any other crossover event
in the same generation, at least within the genomic region considered.
Surprisingly, the corresponding 
dynamics in \emph{continuous} time can be solved in closed form~\citep{MB,reco}.
Again, a crucial ingredient is a transformation to certain correlation
functions (or linkage disequilibria) that linearise and diagonalise the system.
Luckily, in this case,  
the corresponding coefficients are  independent
of the recombination parameters, and the transformation
is   available explicitly.

Motivated by this result, we now investigate 
the analogous single-crossover dynamics in discrete time.
The paper is organised as follows.
We first describe the discrete-time model and the general 
framework (Section~\ref{sec:prelim}) and then recapitulate the essentials 
of the \emph{continuous-time} model and its solution (Section~\ref{SCRconti}). Section~\ref{SCRdiscrete}
returns to discrete time. We first analyse explicitly the cases of two,
three, and four sites. For two and three sites,
the dynamics is analogous to that in continuous time (and, in
particular, available in closed form), but  differs from
then on. This is because  a certain linearity  present in  continuous
time is now lost. The transformations used in 
continuous time are therefore not sufficient to both linearise \emph{and}
diagonalise the discrete-time dynamics. They do, however, lead to
a \emph{linearisation}; this is worked out in Sections~\ref{sec:subsys}~and~\ref{sec:comm}.
The resulting linear system has a triangular structure that
can be diagonalised in a second step in a recursive way (Section~\ref{sec:diagonal}).
We summarise and discuss our results in Section~\ref{sec:discuss}.
An explicit example is worked out in the Appendix.

\section{Preliminaries and notation}\label{sec:prelim}

Let us briefly recall  the recombination model described
in \cite{reco} and the special notation introduced there, as the remainder
of this paper critically depends on it. 
A chromosome (of length $n+1$, say) is represented as a linear arrangement of
the $n+1$ \emph{sites} of the set $S=\lbrace 0, 1 , \ldots , n \rbrace$.
Sites are discrete positions on a chromosome that may be interpreted as
gene or nucleotide positions.
A set $X^{}_i$ collects the possible elements (such as alleles or nucleotides)
at site $i$. For convenience, we restrict ourselves to \emph{finite} sets $X^{}_i$
in this paper, though much of the theory can be extended to the case that
each $X^{}_i$ is a locally compact space, which can be of importance for applications 
in quantitative genetics.
A \emph{type} is now defined as a sequence 
$(x^{}_0, x^{}_1 ,\ldots ,x^{}_n)\in X^{}_0\times X^{}_1\times\cdots\times X^{}_n =:X$,
where $X$ denotes the (finite) \emph{type space}.

Recombination events take place at the so-called \emph{links} between neighbouring sites,
collected into the set $L=\lbrace{\frac{1}{2}, \frac{3}{2}, \ldots ,\frac{2n-1}{2}\rbrace}$,
where link $\alpha=\frac{2i+1}{2}$ is the link between sites $i$ and $i+1$.
Since we only consider single crossovers here, each individual event yields an exchange 
of the sites either before or after the respective link between the two types involved.
A recombination event at link $\frac{2i+1}{2}$ that involves
$x = (x^{}_0, \ldots, x^{}_n)$ and $y=(y^{}_0, \ldots, y^{}_n)$ thus results in the types 
$(x^{}_0, \ldots, x^{}_i, y^{}_{i+1}, \ldots, y^{}_n)$ 
and $(y^{}_0, \ldots, y^{}_i, x^{}_{i+1}, \ldots, x^{}_n)$, with both pairs  
considered as \emph{unordered}.

Although one is ultimately interested in the stochastic process 
defined by recombination acting on populations of finite size, 
compare~\cite{baakeherms} and references therein,
we restrict ourselves to the deterministic limit of
infinite population size here, also known as 
\emph{infinite population limit} (IPL). Consequently, we are
not looking at the individual dynamics, but at the induced dynamics on the
probability distribution on the type space $X$. Let $\mathcal{P}(X)$ denote
the convex space of all possible probability distributions on $X$. 
As $X$ is finite, a probability distribution can be written as a vector
$p=\left(p(x)\right)_{x\in X} $, where $p(x)$ denotes the relative
frequency of type $x$ in the population.

Let us look at the time evolution of the relative frequencies $p^{}_{t} (x)$ of types 
$x=(x^{}_0, \ldots, x^{}_n)$ when starting from  a known initial distribution $p^{}_{0}$ 
of the population at time $t=0$. 
In discrete time, it is given by the following collection of
\emph{recombination equations} for all $x \in X$:
\begin{equation}\label{eq:rekombigleichdis}
  \begin{split}
    p^{}_{t+1} ( x )&= 
     \sum_{\alpha\in L} \rho^{}_{\alpha} \ts
      p^{}_t (x^{}_0 , x^{}_1 , \ldots , x^{}_{\left\lfloor \alpha\right\rfloor}, * , * , \ldots , * ) \,
      p^{}_t ( * , * , \ldots , * , x^{}_{\left\lceil\alpha\right\rceil} ,
      x^{}_{\left\lceil\alpha\right\rceil+1} , \ldots ,x^{}_n) \\
      & \phantom{=}+ \Bigl( 1 - \sum_{\alpha\in L} \rho^{}_\alpha \Bigr) 
      p^{}_t ( x ) \ts , \quad\text{with $t\in \mathbb{N}_0$}\ts ,
  \end{split}
\end{equation}
where the coefficients $\rho^{}_{\alpha}$, $\alpha\in L$, are the
probabilities for a crossover at link $\alpha$. Consequently, we must
have $\rho^{}_{\alpha} \ge  0$ and $\sum_{\alpha\in L} \rho^{}_\alpha \leq 1$, 
where $\rho^{}_{\alpha} > 0$ is assumed from now on without loss of
generality (when $\rho^{}_{\alpha} = 0$, the set 
$X^{}_{\alpha - \frac{1}{2}} \times X^{}_{\alpha + \frac{1}{2}}$ can be considered
as a space for an effective site that comprises $i = \alpha - \frac{1}{2}$
and $i = \alpha + \frac{1}{2}$). When the $\rho^{}_{\alpha}$ do not sum to $1$, the remainder
is the probability that no crossover occurs, which is taken care of by
the last term in the equation. Moreover, $\lfloor\alpha\rfloor$ ($\lceil\alpha\rceil$) 
denotes the largest integer below (the smallest above) $\alpha$ and the star 
$*$ at site $i$ stands for $X^{}_i$, and thus indicates marginalisation over site $i$.

An important step to solve the large nonlinear coupled system of 
equations \eqref{eq:rekombigleichdis} lies in its reformulation in a more 
compact way with the help of certain recombination operators. To
construct them, we need the canonical projection operator 
$\pi^{}_{i} \! : X \longrightarrow X_{i}$, defined by
$x\mapsto \pi^{}_{i} (x) = x_i$ as usual.
Likewise, for any index set $J\subseteq S$, the projector $\pi^{}_J$ is defined as 
$\pi^{}_J \! : X \longrightarrow X^{}_J := \btimes_{i\in J}X^{}_i$.
We will frequently use  
\[
   \pi^{}_{<\alpha} :=   \pi^{}_{\lbrace0,\ldots,\left\lfloor \alpha\right\rfloor\rbrace} 
   \quad\text{and}\quad
   \pi^{}_{>\alpha} := \pi^{}_{\lbrace\left\lceil \alpha\right\rceil,\ldots,n\rbrace}\ts .
\]
These can be understood as \emph{cut-and-forget} operators since 
they `cut out'  the leading and the trailing segment of a type $x$,
respectively, and `forget'  about the rest.
The projectors induce  linear mappings from $\mathcal{P}(X)$ to $\mathcal{P} (X^{}_J) $
by $p \mapsto \pi^{}_{J\cdot}p := p\circ\pi^{-1}_{J}$, where $\pi^{-1}_{J}$ 
denotes the preimage under $\pi^{}_{J}$ and $\circ$ indicates composition of mappings.
The operation $.$ (not to be confused with a multiplication sign) is known 
as the pullback of $\pi^{}_{J}$ with respect to $p$.
 Consequently, $\pi^{}_{J\cdot}p$ is simply the marginal distribution of $p$ 
with respect to the sites of $J$. 

Now consider recombination at link $\alpha$, performed on the entire population.
Since the resulting population consists of randomly chosen leading segments relinked
with randomly chosen trailing segments, it may be described
through the (elementary) recombination operator
(or \emph{recombinator} for short) $R^{}_{\alpha} \! : \mathcal{P}(X) \longrightarrow 
\mathcal{P}(X)$, defined by $p\mapsto R^{}_{\alpha} (p)$ with
\begin{equation}\label{rekombinator}
    R^{}_{\alpha} (p) \, := \, (\pi^{}_{<\alpha\cdot}p) \otimes (\pi^{}_{>\alpha\cdot} p) \ts ,
\end{equation}
where $\otimes$ denotes the product measure and reflects the independent combination 
of both marginals $\pi^{}_{<\alpha\cdot}p$ and $\pi^{}_{>\alpha\cdot} p$.
Note that the recombinators are structural operators that  do not
depend on the recombination probabilities.

Before we rewrite the recombination equations in terms of these recombinators,
let us recall some of their elementary properties, see~\cite{reco} for proofs.
First of all, the elementary recombinators $R^{}_{\alpha}$ are idempotents 
and commute with one another on $\mathcal{P} (X)$. This permits the
consistent definition of \emph{composite recombinators}
\begin{equation} \label{eq:defcomposite}
    R^{}_{G} \, := \, \prod_{\alpha\in G} R^{}_{\alpha}
\end{equation}
for arbitrary subsets $G\subseteq L$.  In particular, one has $R^{}_{\varnothing} = \one$
and $R^{}_{ \{ \alpha\} } = R^{}_{\alpha}$.
\begin{prop}\label{recoprop}
    On $\mathcal{P} (X)$, the elementary
    recombinators are commuting idempotents.
For $\alpha\leq\beta$, they satisfy
\begin{equation}\label{elemprop1}
    \pi^{}_{<\alpha\cdot} \bigl( R^{}_{\beta} (p) \bigr) 
        = \pi^{}_{<\alpha\cdot} p   \quad \text{ and} \quad
       \pi^{}_{>\alpha\cdot}  \bigl( R^{}_{\beta} (p) \bigr)
        = (\pi^{}_{\{ \lceil\alpha\rceil,\ldots,\lfloor\beta\rfloor \}\cdot}p)\otimes (\pi^{}_{>\beta\cdot}p) \, ;
\end{equation}
likewise, for $\alpha\geq\beta$,
\begin{equation}\label{elemprop2}
 \pi^{}_{>\alpha\cdot}\bigl(R^{}_{\beta}(p)\bigr)  =\pi^{}_{>\alpha\cdot}p \quad\text{ and} \quad
 \pi^{}_{<\alpha\cdot} \bigl( R^{}_{\beta} (p) \bigr) 
  =(\pi^{}_{<\beta\cdot}p) \otimes (\pi^{}_{\{ \lceil\beta\rceil,\ldots,\lfloor\alpha\rfloor \}\cdot}p) \,.
\end{equation}

      Furthermore, the composite recombinators  satisfy
      \begin{equation}\label{composition}
        R^{}_{G} R^{}_{H} = R^{}_{G\cup H}
      \end{equation}
  for arbitrary $G, H\subseteq L$. \qed
\end{prop}
These properties can be understood intuitively as well:
\eqref{elemprop1} says that recombination at or after link $\alpha$ does not affect the
marginal frequencies at sites before $\alpha$, whereas the marginal frequencies at the
sites after $\alpha$ change into the product measure 
(and vice versa in~\eqref{elemprop2}).
Furthermore, repeated recombination at link $\alpha$ does not change the situation any further
(recombinators are idempotents) and the formation of the product measure with respect to
$\geq 2$ links does not depend on the order in which the links are affected.
As we shall see below, these properties of the recombinators are crucial for 
finding a solution of the SCR dynamics, both in continuous and in discrete time.

\section{SCR in continuous time}\label{SCRconti}
Let us briefly review the SCR dynamics in \emph{continuous} time, as its structure will be
needed below. Making use of the recombinators introduced above, the dynamics (in the IPL)
is described by a system of differential equations for the time evolution of the probability
distribution (or measure), starting from an initial condition $p^{}_{0}$ at $t=0$. It reads~\cite{reco}
\begin{equation} \label{rekombiglstet}
   \dot{p}^{}_{t} \ts =  \sum_{\alpha\in L}\widetilde{\rho}^{}_{\alpha} \ts
                            \bigl(R^{}_{\alpha}-\one\bigr) (p^{}_{t}) \ts ,
\end{equation}
where $\widetilde{\rho}^{}_{\alpha}$ is now the \emph{rate} for a crossover at link $\alpha$.
Though \eqref{rekombiglstet} describes a coupled system of nonlinear 
differential equations, the closed solution for its Cauchy (or initial value)
problem is available~\citep{MB,reco}:

\begin{theorem}  \label{thm:reco-cont}
   The solution of the recombination equation $\eqref{rekombiglstet}$ with initial 
   value $p^{}_0$ can be given in closed form as
\begin{equation} \label{solution}
   p^{}_{t} \ts =  \sum_{G\subseteq L}\widetilde{a}^{}_{G} (t) R^{}_{G} (p^{}_{0})\ts ,
\end{equation}
   with the coefficient functions
\begin{equation} \label{koefffunk}
    \widetilde{a}^{}_{G} (t) \, = \!
    \prod_{\alpha\in L\setminus G}\!\! \exp(-\widetilde{\rho}^{}_{\alpha}t)
    \prod_{\beta\in G}(1-\exp(-\widetilde{\rho}^{}_{\beta} t))\ts .
\end{equation}
   These are non-negative functions, which satisfy
   $ \sum_{G\subseteq L} \widetilde{a}^{}_{G} (t)=1$ for all $t\ge 0$.
   \qed
\end{theorem}
The coefficient functions  can be interpreted probabilistically. Given an
individual sequence in the population, $\widetilde{a}^{}_{G} (t)$ is the 
probability that the set of links that have seen at least one crossover 
event until time $t$ is precisely the set $G$.
Note that the product structure of the $\widetilde{a}^{}_{G} (t)$
implies independence of links, a decisive feature of the single-crossover
dynamics in continuous time, as we shall see later on.
By~\eqref{solution}, $p^{}_t$ is always a convex combination of the probability 
measures $R^{}_G(p^{}_0)$ with $G\subseteq L$.
Consequently, given an initial condition $p_0^{}$, the entire dynamics takes 
place on the closed simplex (within $\mathcal{P}(X)$) that is given by 
${\rm conv}\{R_G^{}(p^{}_0) \mid G\subseteq L\}$,
where ${\rm conv}(A)$ denotes the convex hull of $A$.

It is surprising that a closed solution for the dynamics~\eqref{rekombiglstet} can 
be given explicitly, and this suggests the existence of an underlying linear structure~\cite{MB}, 
which is indeed the case and well known from similar equations, compare~\cite{Lyubich}. 
In the context of the formulation with recombinators, it can be stated as 
follows, compare~\cite{reco} for details.

\begin{theorem} \label{thm:linear}
    Let $\bigl\{c_{G'}^{(L')}(t) \mid \varnothing\subseteq G'\subseteq L'\subseteq L\bigr\}$ 
    be a family of non-negative functions with 
    $c_G^{(L)}(t) = c_{G_1}^{(L_1)}(t)\, c_{G_2}^{(L_2)}(t)$, valid for any 
    partition $L=L_1\ts\dot{\cup}\ts L_2$ of the set $L$ and all $t\geq 0$,
    where $G^{}_i := G\cap L^{}_i$.
    Assume further that these
    functions satisfy $\sum_{H\subseteq L'} c_{H}^{(L')} (t) = 1$ for any 
    $L'\subseteq L$ and $t\geq 0$. 
    If $v\in\mathcal{P}(X)$ and $H\subseteq L$, one has the identity
\[
     R^{}_{H} \Bigl(\sum_{G\subseteq L} c^{(L)}_{G} (t) R^{}_{G} (v) \Bigr) 
     \, =  \sum_{G\subseteq L} c^{(L)}_{G} (t) R^{}_{G\cup H} (v) \ts ,
\]     
     which is then satisfied for all $ t\ge 0$.   \qed
\end{theorem}
Here, the upper index specifies the respective set of links. 
So far, Theorem~\ref{thm:linear} depends crucially on the product
structure of the functions $c_G^{(L)}(t)$, but we will show later how
this assumption can be relaxed.
In any case, the coefficient functions
$\widetilde{a}^{}_{G} (t)$ of \eqref{koefffunk} satisfy the conditions of 
Theorem~\ref{thm:linear}. The result then means that the recombinators 
act linearly along solutions~\eqref{solution} of the recombination equation~\eqref{rekombiglstet}.
Denoting $\varphi^{}_t$ as the flow of Eq.~\eqref{rekombiglstet},
Theorem~\ref{thm:linear} thus has the following consequence.

\begin{coro}\label{coro1}
   On $\mathcal{P}(X)$, the forward flow of $\eqref{rekombiglstet}$ commutes with all 
   recombinators, which means that $R^{}_G\circ \varphi^{}_t  = \varphi^{}_t \circ R^{}_G$ holds for all 
   $t\geq 0$ and all $G\subseteq L$.          \qed
\end{coro}

The conventional approach to solve the recombination dynamics consists in transforming 
the type frequencies to certain functions
which diagonalise the dynamics, see 
\citep{Bennett, Kevin1, Kevin2, Lyubich} and references therein for more.
From now on, we will call these functions \emph{principal components} 
after Bennett~\cite{Bennett}.
For the single-crossover dynamics in continuous time, they have a particularly
simple structure: they are given by certain correlation functions,
or {\em linkage disequilibria} (LDE), which
play an important role in biological applications.
They have a counterpart at the level of operators on $\mathcal{P}(X)$.

Namely, let us define \emph{LDE operators} on $\mathcal{P}(X)$ 
as linear combinations of recombinators via
\begin{equation} \label{trafo}
    T^{}_{G} := \sum_{H\supseteq G}(-1)^{\lvert H-G \rvert } R^{}_{H}\ts ,\quad \mbox{with } G\subseteq L \, ,
\end{equation}
so that the inverse relation is given by
\begin{equation} \label{trafo-2}
     R^{}_{H} \, = \sum_{G\supseteq H} T^{}_{G} 
\end{equation}
due to the combinatorial M\"{o}bius inversion formula, compare~\cite{Aigner}.
Let us note for further use that, by Eq.~\eqref{composition} in Proposition~\ref{recoprop},
$T^{}_G \circ R^{}_G = T^{}_G$. Note also that, for
a probability measure $p$ on $X$, $T^{}_G (p)$ is a signed measure on $X$;
in particular, it need not be positive.
The LDEs are given by certain {\em components} of the $T^{}_G (p)$ 
--- see~\citep{reco,baakeherms} for more.
In the continuous-time single-crossover setting,
it was shown in \cite{reco} that, if $p^{}_{t}$ is the solution~\eqref{solution},
the $T^{}_{G} (p^{}_{t})$ satisfy
\begin{equation}
   \frac{\mathrm{d}}{\mathrm{d} t} \, T^{}_{G} (p^{}_{t})
   \, = \, - \Bigl(\sum_{\alpha\in L\setminus G}\widetilde{\rho}^{}_{\alpha}\Bigr)
   T^{}_{G} (p^{}_{t}) \ts , \quad  \mbox{for all } G\subseteq L \ts , 
\end{equation}
which is a \emph{decoupled} system of homogeneous linear differential equations, with
the standard exponential solution.
That is, the LDE operators both linearise and diagonalise the system, and 
the LDEs are thus, at the same time, principal components.

A straightforward calculation now reveals that the solution~\eqref{solution}
can be rewritten as
\begin{equation} \label{two-versions}
   p^{}_{t} \, = \sum_{G\subseteq L} \widetilde{a}^{}_{G} (t) R^{}_{G} (p^{}_{0})
   \, = \sum_{K\subseteq L} b^{}_{K} (t) \ts T^{}_{K} (p^{}_{0})
   \, = \sum_{K\subseteq L} T^{}_{K} (p^{}_{t}) \ts ,
\end{equation}
where the new coefficient functions are given by
\[
    b^{}_{K} (t) :=\, \exp\Bigl(- \!\!\! \sum_{\alpha\in L\setminus K}
           \widetilde{\rho}^{}_{\alpha} t \Bigr).
 \]

At this point, it is important to notice the rather simple structure of the LDE operators,
which do not depend on the crossover rates. Moreover, the transformation between
recombinators and LDE operators is directly given by the M\"obius formula, see
Eqs.~\eqref{trafo} and \eqref{trafo-2}. This is a significant simplification in
comparison with previous results, compare \citep{Bennett, Kevin1, Kevin2, Geiringer},
where the coefficients of the transformation generally depend on the crossover rates 
and must be determined recursively.

Below, we shall see that the SCR dynamics in \emph{continuous} time is indeed 
a special case, and that the above results cannot be transferred directly to the 
corresponding dynamics in \emph{discrete} time. Nevertheless, part of the 
continuous-time structure prevails and offers a useful entry point 
for the solution of the discrete-time counterpart.

\section{SCR in discrete time} \label{SCRdiscrete}

Employing recombinators, the SCR equations~\eqref{eq:rekombigleichdis} 
in discrete time with
a given initial distribution $p^{}_0$ can be compactly rewritten as
\begin{equation} \label{rekooper}
    p^{}_{t+1} \ts = \ts p^{}_t + \sum_{\alpha\in L} \rho^{}_{\alpha} \bigl( R_{\alpha}
        - \one \bigr) ( p^{}_t )  =: \varPhi(p^{}_t) \, .
\end{equation}
As indicated, the nonlinear operator of the right-hand side 
of \eqref{rekooper} is denoted by $\varPhi$ from now on.
We aim at a closed solution of \eqref{rekooper}, namely for 
$p^{}_t = \varPhi^{t} ( p^{}_0 )$ with $ t\in \mathbb{N}_0$.
Based on the result for the continuous-time model,
the solution is expected to be of the form

\begin{equation}\label{solutiondiscrete}
   p^{}_t \,  = \, \varPhi^t  (p^{}_0 ) 
       \ts  = \sum_{G\subseteq L} a^{}_G ( t ) R^{}_G ( p^{} _0 ) \, ,
\end{equation}
with non-negative $a^{}_G(t)$, $G\subseteq L$, $\sum_{G\subseteq L} a^{}_G ( t ) = 1$,
describing the (unknown) coefficient functions arising from the dynamics.
This representation of the solution was first stated by Geiringer~\cite{Geiringer}.
In particular, also the discrete-time dynamics takes place on the simplex
${\rm conv}\{ R^{}_G(p^{}_0) \mid G\subseteq L\}$.

We are particularly interested in whether a discrete-time equivalent
to Corollary~\ref{coro1} exists, that is, whether all recombinators
$R^{}_G$ commute with $\varPhi$.
This is of importance since it would allow for a diagonalisation
of the dynamics via the LDE operators~\eqref{trafo}. To see this,
assume for a moment that $R^{}_{\alpha}\circ\varPhi = \varPhi\circ R^{}_{\alpha}$ for all
$\alpha\in L$, and thus $R^{}_G\circ\varPhi = \varPhi\circ R^{}_G$ 
for all $G\subseteq L$.
Noting that, when $\alpha~\in ~G~\subseteq H$, Eq.~\eqref{composition} 
from Proposition~\ref{recoprop}  implies that
$(R^{}_{\alpha} - \one)R^{}_H  = R^{}_{H\cup\{\alpha\}} - R^{}_H = R^{}_H - R^{}_H = 0$,
we see that the assumption above would lead to
\begin{eqnarray*}
\lefteqn{T^{}_G \circ \varPhi \, = \, \sum_{H\supseteq G}(-1)^{\vert H - G\vert} R^{}_H \circ \varPhi
           \, = \, \sum_{H\supseteq G}(-1)^{\vert H - G\vert} \varPhi \circ R^{}_H } \\[1mm]
           &\, = \,&  \sum_{H\supseteq G}(-1)^{\vert H - G\vert} R^{}_H 
             + \sum_{H\supseteq G} (-1)^{\vert H - G\vert} 
                  \sum_{\alpha \in L}\rho^{}_{\alpha} (R^{}_{\alpha} - \one) R^{}_H \\[1mm]
           & \, = \, & T^{}_G + \sum_{H\supseteq G} (-1)^{\vert H - G \vert} 
                  \sum_{\alpha \in L\setminus G}\rho^{}_{\alpha} (R^{}_{\alpha} - \one) R^{}_H \\
          &\, = \, & \Bigl(1- \sum_{\alpha \in L\setminus G}\rho^{}_{\alpha}\Bigr) T^{}_G
              + \sum_{\alpha \in L\setminus G}\rho^{}_{\alpha} 
              \sum_{\substack{H \supseteq G\\\alpha\notin H}}
              \Bigl( (-1)^{\vert H - G\vert} R^{}_{H \cup \{\alpha\}} +
              (-1)^{\vert H \cup \{\alpha\} - G \vert} R^{}_{H\cup \{\alpha\}} \Bigr) \\
          & \, = \,& \Bigl(1- \sum_{\alpha \in L\setminus G} \rho^{}_{\alpha}\Bigr) T^{}_G \, ,        
\end{eqnarray*}
so that, indeed, all $T^{}_G(p^{}_t)$ would decay geometrically.
This wishful calculation is badly smashed by the nonlinear nature of the
recombinators, and the remainder of this paper is concerned with true identities
that repair the damage.

To get an intuition for the dynamics in discrete time, let us first take a closer
look at the discrete-time model with two, three, and four sites.

\subsection{Two and three sites}
For two sites, one simply has $S = \lbrace0,1\rbrace$ and $L = \lbrace\frac{1}{2}\rbrace$,
so that only one non-trivial recombinator exists, $R  = R^{}_{\frac{1}{2}}$, with corresponding
recombination probability $\rho = \rho^{}_{\frac{1}{2}}$.
Consequently, the SCR equation simplifies to
\begin{equation}\label{eq2sites}
   p^{}_{t+1} \, = \, \varPhi \ts (p^{}_{t})
   \, = \, \rho \, R (p^{}_{t}) + (1-\rho) \, p^{}_{t} \ts ,
\end{equation}
where $p^{}_{t}$ is a $\vert X\vert$-dimensional probability vector. The
solution is given by
\begin{equation}\label{sol2sites}
   p^{}_{t} \, = \, a (t) \, p^{}_{0} +  \bigl( 1 - a(t) \bigr) \ts R (p^{}_{0})
\end{equation}
with $a(t) = a^{}_{\varnothing} (t) = (1-\rho)^t$.
This formula is easily verified by induction \cite{Ute}. Thus, in analogy with the SCR 
dynamics in continuous time, the solution is available 
in closed form, and the coefficient functions allow an analogous probabilistic 
interpretation. Furthermore, it is easily seen that
the recombinators $R^{}_{\varnothing} = \one$ and $R^{}_{\frac{1}{2}} = R$ commute
with $\varPhi$ and therefore with $\varPhi^t$ for all $t\in \mathbb{N}_0$.
For two sites, the analogue of Corollary~\ref{coro1} thus holds 
in discrete time.
As a consequence, the LDE operators 
from \eqref{trafo} decouple and linearise the system~\eqref{eq2sites}.
At the level of the component LDEs, this is common knowledge
in theoretical population genetics; compare~\cite[Chap.3]{Hartl}.

\smallskip
Similarly, the recombination equation~\eqref{eq:rekombigleichdis} for three sites can be 
solved explicitly as well. An elementary calculation (applying the iteration and 
comparing coefficients) shows that the corresponding coefficient functions 
$a^{}_{G} (t)$ follow the linear recursion
\begin{eqnarray*}
\left(
\begin{array}{c}
a^{}_{\varnothing}(t+1)\\
a^{}_{\frac{1}{2}}(t+1)\\
a^{}_{\frac{3}{2}}(t+1)\\
a^{}_{\left\{\frac{1}{2},\frac{3}{2}\right\}}(t+1)
\end{array}
\right)
&=&
\begin{pmatrix}
1-\rho^{}_{\frac{1}{2}}-\rho^{}_{\frac{3}{2}}&0&0&0\\
\rho^{}_{\frac{1}{2}}&1-\rho^{}_{\frac{3}{2}}&0&0\\
\rho^{}_{\frac{3}{2}}&0&1-\rho^{}_{\frac{1}{2}}&0\\
0&\rho^{}_{\frac{3}{2}}&\rho^{}_{\frac{1}{2}}&1
\end{pmatrix}
\left(
\begin{array}{c}
a^{}_{\varnothing}(t)\\
a^{}_{\frac{1}{2}}(t)\\
a^{}_{\frac{3}{2}}(t)\\
a^{}_{\left\{\frac{1}{2},\frac{3}{2}\right\}}(t)
\end{array}
\right) \, ,
\end{eqnarray*}
with solution
\begin{equation}\label{solthreesites}
 \begin{split}
   a^{}_{ \varnothing } (t) &\, = \, \bigl(1 - \rho^{}_{\frac{1}{2}} - \rho^{}_{\frac{3}{2}}\bigr)^{t} \, , \\
  a^{}_{ \frac{1}{2}} (t)  &\, = \, \bigl(1 - \rho^{}_{\frac{3}{2}}\bigr)^{t} - 
         \bigl(1 - \rho^{}_{\frac{1}{2}} - \rho^{}_{\frac{3}{2}}\bigr)^{t} \, , \\
  a^{}_{\frac{3}{2}} (t) &\,=\, \bigl(1 - \rho^{}_{\frac{1}{2}}\bigr)^{t} - 
         \bigl(1 - \rho^{}_{\frac{1}{2}} - \rho^{}_{\frac{3}{2}}\bigr)^{t} \, , \\
  a^{}_{\left\{\frac{1}{2},\frac{3}{2}\right\}} (t) &\,= \,
        1  - \bigl(1 - \rho^{}_{\frac{3}{2}}\bigr)^{t} - 
        \bigl(1 - \rho^{}_{\frac{1}{2}}\bigr)^{t} 
        + \bigl(1 - \rho^{}_{\frac{1}{2}} - \rho^{}_{\frac{3}{2}}\bigr)^{t} \, .
 \end{split}
\end{equation}
These coefficient functions have the same probabilistic interpretation as the corresponding
$\widetilde{a}^{}_G(t)$, $G\subseteq L$, in the continuous-time model,
so that $a^{}_G(t)$ is the probability that the links that have been involved
in recombination until time $t$ are exactly those of the set $G$.

But there is a crucial difference.
Recall that, in continuous time, single crossovers imply
\emph{independence} of links, which is expressed in the product structure of the 
coefficient functions $\widetilde{a}^{}_G(t)$ (see~\eqref{koefffunk}).
This independence is lost in discrete time, where a crossover event at one link 
forbids any other cut at other links in the same time step.
Consequently, already for three sites, the coefficients of the discrete-time dynamics
fail to show the product structure used in Theorem~\ref{thm:linear}.

But even though Corollary~\ref{coro1}, concerning the forward flow of
\eqref{rekombiglstet}, is a consequence of Theorem~\ref{thm:linear}, which, in turn, 
is based upon the product structure of the coefficients, 
a short calculation reveals that $R^{}_G\circ\varPhi =\varPhi \circ R^{}_G$ 
still holds for the discrete-time
model with three sites for all $G \subseteq \bigl\{\frac{1}{2}, \frac{3}{2}\bigr\}$.
As a consequence, just as in the case of two sites, 
the $T^{}_G$ linearise {\em and} decouple the 
dynamics, which is well-known to the experts, see~\citep{Bennett, Christiansen} for more.

To summarise: despite the loss of independence of links, an
explicit solution of the discrete-time recombination dynamics is still available,
and a linearisation and 
diagonalisation of the dynamics can be achieved with the methods 
developed for the continuous-time model,
that is, a transformation to a solvable system via the 
$T^{}_G$. However, things will 
become more complex if we go to four sites and beyond. 
In particular, there is no equivalent to Corollary~\ref{coro1},
i.e., in general, the recombinators {\em do not} commute with $\varPhi^{}$,
and we have to search for a new transformation that replaces \eqref{trafo},
as will be explained next.

\subsection{Four sites}
The complication with four sites originates from the fact that 
$R^{}_{\frac{3}{2}}\circ \varPhi \neq \varPhi \circ R^{}_{\frac{3}{2}}$,
so that the property described by Corollary~\ref{coro1} for continuous time is lost here.
Consequently, the $T^{}_G$
fail the desired properties. In particular, one finds
\[
    T^{}_{\varnothing} (\varPhi(p)) \, = \, 
    \bigl(1 - \rho^{}_{\frac{1}{2}} - \rho^{}_{\frac{3}{2}} - \rho^{}_{\frac{5}{2}}\bigr) \ts
     T^{}_{\varnothing}(p) - \rho^{}_{\frac{1}{2}} \, \rho^{}_{\frac{5}{2}} \,T^{}_{\frac{3}{2}} (p) \ts ,
\]
so that an explicit solution of the model cannot be obtained as before. 

This raises the question why four sites are more difficult than three sites,
even though independence of links has already been lost
with three sites.
To answer this, we look at the time evolution of the coefficient functions
$a^{}_G(t)$, $G\subseteq L$.
For this purpose, let us return to the general model with an arbitrary number of sites.
\subsection{General case}
We now consider an arbitrary (but finite) set $S$ with the corresponding link set $L$.
For each $G\subseteq L$, we use the following abbreviations:
\begin{equation*}
  \begin{split}
    G^{}_{< \alpha} &:= \left \{ i \in G \mid i < \alpha \right \} , \quad
    G^{}_{> \alpha}  :=  \left \{ i \in G \mid i > \alpha \right \} ,\\
    L^{}_{ \leq \alpha} &:= \left \{ i \in L \mid i \leq \alpha \right \} , \quad
    L^{}_{ \geq \alpha}  := \left \{ i \in L \mid i \geq \alpha \right \}. 
  \end{split}
\end{equation*} 
Furthermore, we set $\eta := 1- \textstyle\sum_{\alpha \in L} \rho^{}_{\alpha}$.
We then obtain
\begin{theorem}\label{thm:adevelop}
  For all $G\subseteq L$ and $t\in\mathbb{N}_0$, the coefficient functions $a^{}_G(t)$ evolve according to
   \begin{equation}\label{nonlinrecur}
     a^{}_G(t+1)\, = \, \eta \thinspace a^{}_G(t) + \sum_{\alpha \in G } \rho^{}_{ \alpha} 
          \Bigl( \sum_{H \subseteq L_{ \geq \alpha} }  a^{}_{G^{}_{ < \alpha}\cup H}(t) \Bigr) \thinspace
          \Bigl( \sum_{K \subseteq L_{ \leq \alpha }} a^{}_{K \cup G^{}_{> \alpha} }(t) \Bigr) \ts ,
   \end{equation}
with initial condition $a^{}_G(0) = \delta^{}_{G,\varnothing}$.
\end{theorem}

 \begin{proof}
Geiringer \cite{Geiringer} already explained in words how to derive this general
recursion, and illustrated it with the four-site example;
we give a proof via our operator formalism.
  Using \eqref{solutiondiscrete}, the recombination equation for $p^{}_{t+1}$ reads
  \begin{equation*}
   \begin{split}
p^{}_{t+1} &\, = \, \sum_{ G \subseteq L } a^{}_G (t + 1) R_G ( p^{}_0 ) =
   \varPhi \bigl(p^{}_t \bigr) = \varPhi \biggl( \sum_{ G \subseteq L } a^{}_G (t) R^{}_G (p^{}_0) \biggr)\\
        &\, =\, \sum_{ \alpha \in L } \rho^{}_{ \alpha } \thinspace
     \biggl(\Bigl( \pi^{}_{ < \alpha \cdot } 
     \Bigl( \sum_{ H \subseteq L } a^{}_H (t) R^{}_H (p^{}_0) \Bigr)\Bigr)
         \otimes \Bigl(\pi^{}_{ > \alpha \cdot }
        \Bigl( \sum_{ K \subseteq L } a^{}_K (t) R^{}_K (p^{}_0) \Bigr)\Bigr) \biggr) \\
        &\phantom{=}+ \eta \thinspace \Bigl( \sum_{ G \subseteq L } a^{}_G (t) R^{}_G (p^{}_0) \Bigr) \, ,
   \end{split}
  \end{equation*}
  where each product term in the first sum can be calculated as
  \begin{equation*}
\begin{split}
    \biggl( \pi^{}_{ < \alpha \cdot } 
     \Bigl( \sum_{ H \subseteq L } & a^{}_H (t) R_H (p^{}_0) \Bigr)\biggr)
         \otimes \biggl(\pi^{}_{ > \alpha \cdot }
        \Bigl( \sum_{ K \subseteq L } a^{}_K (t) R^{}_K (p^{}_0) \Bigr) \biggr)\\
     &= \sum_{ H, K \subseteq L } a^{}_H (t) a^{}_K (t)
         \biggl( \Bigl(\pi^{}_{ < \alpha \cdot } R^{}_H (p^{}_0)\Bigr)
             \otimes \Bigl(\pi^{}_{ > \alpha \cdot } R^{}_K (p^{}_0) \Bigr)\biggr) \\
     &= \sum_{ H , K \subseteq L } a^{}_H (t) a^{}_K (t)
         \biggl(\Bigl( \pi^{}_{ < \alpha \cdot } R^{}_{ H^{}_{ < \alpha } 
                            \cup K^{}_{ > \alpha } } (p^{}_0) \Bigr)
         \otimes \Bigl(\pi^{}_{ > \alpha \cdot } R^{}_{ H_{ < \alpha } 
                                \cup K^{}_{ > \alpha } } (p^{}_0) \Bigr)\biggr) \\
     &= \sum_{ H , K \subseteq L }a^{}_H (t) a^{}_K (t)
         \biggl( R^{}_{ \alpha } \Bigl( R^{}_{ H^{}_{ < \alpha } \cup K^{}_{ > \alpha } } (p^{}_0) \Bigr) \biggr) \, ,
\end{split}
    \end{equation*}
where we use the linearity of the projectors in the first step,
and Eqs.~\eqref{elemprop1} and \eqref{elemprop2} from Proposition~\ref{recoprop}  
in the second (more precisely, we use the left parts of Eqs.~\eqref{elemprop1} and \eqref{elemprop2},
reading them both forward and backward).
Insert this into the expression for $p^{}_{t+1}$ and rearrange the sums for a comparison
of coefficients of $R^{}_G$ with $G\subseteq L$.
Comparison of coefficients is justified by the observation that, for generic $p^{}_0$
and generic site spaces,
the vectors $R^{}_G(p^{}_0)$ with $G\subseteq L$ are the extremal vectors of the closed simplex
${\rm conv}\{R_K^{}(p^{}_0) \mid K\subseteq L\}$.
They are the vectors that (generically) cannot be expressed as non-trivial convex combination within
the simplex, and hence the vertices of the latter
(in cases with degeneracies, one reduces the simplex in the obvious way).
If $G = \varnothing$, we only have $\eta \thinspace a^{}_{\varnothing} (t)$
as coefficient for $R^{}_{\varnothing}$.
Otherwise, we get additional contributions for each $\alpha\in G$, namely,
from those $H,K\subseteq L$ for which
$H^{}_{ < \alpha } = G^{}_{ < \alpha }$ and $K^{}_{ > \alpha } = G^{}_{ > \alpha }$, 
while $H^{}_{ \geq \alpha }$ and $K^{}_{ \leq \alpha }$ can be any subset of $L^{}_{ \geq \alpha }$
and $L^{}_{ \leq \alpha }$, respectively.
Hence, the term belonging to $R^{}_G (p^{}_0)$ reads 
   \[
       \sum_{ \alpha \in G } \rho^{}_{ \alpha } \Bigl( \sum_{ H \subseteq L_{ \geq \alpha } }
       \sum_{ K \subseteq L_{ \leq \alpha } } a^{}_{ G^{}_{ < \alpha } \cup H } (t)
          \thinspace a^{}_{ K\cup G^{}_{ > \alpha }  } (t) \Bigr) + \eta \thinspace a^{}_G (t) \, ,
   \]
   and the assertion follows. \qed
\end{proof}
The $a^{}_G(t)$ have the same probabilistic interpretation as the 
$\widetilde{a}^{}_G(t)$~\eqref{koefffunk} from the continuous-time model,
and the above iteration can be understood intuitively as well:
A type $x$ resulting from recombination at link $\alpha$ is composed of two segments $x^{}_{<\alpha}$
and $x^{}_{>\alpha}$. These segments themselves may have been pieced together in previous 
recombination events already, and the iteration explains the possible cuts these segments may
carry along. The first term in the product stands for the type delivering the leading segment (which
may bring along arbitrary cuts in the trailing segment),
the second for the type delivering the trailing one (here any leading segment is allowed).
The term $\eta \thinspace a^{}_G(t)$ covers the case of no recombination.

Note that the above iteration is generally {\em nonlinear},
where the products stem from the fact that types 
recombine independently.
This nonlinearity is the reason that an explicit solution
cannot be given as before.

A notable exception is provided by recombination events that
occur at links where one of the involved segments cannot have been affected by previous crossovers,
namely the links $\frac{1}{2}$ and $\frac{2n-1}{2}$. 
In this case,
at least one of the factors in Eq.~\eqref{nonlinrecur} becomes $1$ 
(since, obviously, $G^{}_{<\alpha} = \varnothing$ for $\alpha = \frac{1}{2}$ 
and $G^{}_{>\alpha} = \varnothing$ for $\alpha = \frac{2n-1}{2}$) and the
resulting linear and triangular recursion can be solved.
The coefficients for the corresponding link sets 
can be inferred directly (proof via simple induction) as 
\begin{equation}\label{adirectly}
   \begin{split}
    a^{}_{ \varnothing} ( t ) &\, = \,  \eta^t \, , \\
    a^{}_{ \frac{1}{2}} ( t ) &\, =\, \left( \eta + \rho^{}_{ \frac{1}{2}}\right)^t - \eta^t \, ,  \\
    a^{}_{ \frac{2n-1}{2}} ( t ) &\, =\, \left( \eta + \rho^{}_{\frac{2n-1}{2}} \right)^t - \eta^t \, ,
                                                                    \quad\text{and} \\
    a^{}_{\left\{\frac{1}{2},\frac{2n-1}{2}\right\}}(t) &\,=\,
       \eta^t - \left( \eta + \rho^{}_{\frac{1}{2}} \right)^t -
    \left( \eta + \rho^{}_{\frac{2n-1}{2}} \right)^t +
    \left( \eta + \rho^{}_{\frac{1}{2}} + \rho^{}_{\frac{2n-1}{2}} \right)^t \,.
\end{split}
\end{equation}
This explains the availability of an explicit solution for the model with up to three sites,
where we do not have links other than $\frac{1}{2}$ and/or $\frac{3}{2}$,
so that all corresponding coefficients can be determined explicitly.
Indeed, one recovers \eqref{solthreesites} with $n = 2$ and
$\eta = 1-\rho^{}_{\frac{1}{2}}-\rho^{}_{\frac{3}{2}}$.

So far, we have observed that the product structure of the coefficient functions, known 
from continuous time, is lost in discrete time from three sites onwards;
this reflects the dependence of links. 
In contrast, the linearity of the iteration is only lost from four sites onwards.
The latter can be understood further by comparison of \eqref{nonlinrecur} with
the differential equations for the coefficients of the continuous-time model.
These read:
\begin{equation}\label{diffeqcont1}
  \frac{\mathrm{d}}{\mathrm{d} t} \widetilde{a}^{}_G(t) \, = \,
               -\Bigl(\sum_{\alpha\in L\setminus G}\tilde{\rho}^{}_{\alpha} \Bigr)
                  \widetilde{a}^{}_G(t)
               + \sum_{\alpha\in G} \tilde{\rho}^{}_{\alpha} \ts
                 \widetilde{a}^{}_{G\setminus \{ \alpha\} } (t) \, ,
\end{equation}
that is, they are linear, with solution~\eqref{koefffunk}. 
Note that this linear dynamics emerges from a seemingly nonlinear
one, namely the analogue of~\eqref{nonlinrecur},
\begin{equation}\label{diffeqcont2}
   \frac{\mathrm{d}}{\mathrm{d} t} \widetilde{a}^{}_G(t) \, = \, 
               -\Bigl(\sum_{\alpha\in L}\tilde{\rho}^{}_{\alpha}\Bigr)
                  \widetilde{a}^{}_G(t)
               + \sum_{\alpha\in G} \tilde{\rho}^{}_{\alpha}
       \Bigl( \sum_{H \subseteq L^{}_{ \geq \alpha} }  
             \widetilde{a}^{}_{G^{}_{ < \alpha}\cup H}(t) \Bigr) \,
        \Bigl( \sum_{K \subseteq L^{}_{ \leq \alpha }} 
                    \widetilde{a}^{}_{ K \cup G^{}_{> \alpha} }(t) \Bigr) \ts .
\end{equation}
However, due to the product structure of the solution, the product term in the second sum,
when inserting \eqref{koefffunk},
reduces to a single term,
\[
   \Bigl( \sum_{H \subseteq L_{ \geq \alpha} }  
             \widetilde{a}^{}_{G^{}_{ < \alpha}\cup H}(t) \Bigr) \,
        \Bigl( \sum_{K \subseteq L_{ \leq \alpha }} \widetilde{a}^{}_{K \cup G^{}_{> \alpha}}(t) \Bigr)
  \, = \,   \widetilde{a}^{}_{G}(t) + \widetilde{a}^{}_{G\setminus \{\alpha\}}(t) \, ,
\]
which turns \eqref{diffeqcont2} into \eqref{diffeqcont1}.

What happens here is the following.
From four sites onwards (namely, beginning with $n=3$ and a crossover
at $\alpha = \frac{3}{2}$, and both in discrete and continuous time),
it happens that leading and trailing segments meet that both possess
at least one link that may possibly have seen a previous cut.
When a crossover at $\alpha$ takes place, the new joint distribution
of cuts before and after $\alpha$ is formed as the product measure
of the marginal distributions of cuts in the leading and trailing
segments (cf. \eqref{nonlinrecur} and \eqref{diffeqcont2})
--- akin to the formation of product measures of marginal types
by $R^{}_{\alpha}$.
In continuous time, the links are all independent, hence
the new combination leaves the joint distribution of cuts
unchanged.
Therefore, a set $G$ of affected links (before and after $\alpha$)
is simply augmented by $\alpha$ if $\alpha$ is a `fresh' cut;
this results in the linearity of \eqref{diffeqcont1}.
In discrete time, however, the dependence between the links,
in particular between those in the leading and trailing segment,
means that the formation of the product measure changes the
joint distribution of affected links, in addition to the new
cut at $\alpha$; thus \eqref{nonlinrecur} remains nonlinear.

Since we aim
at an explicit solution of the discrete-time recombination model, 
we need to find a way to overcome the obstacles of nonlinearity.
Inspired by the results of the continuous-time model, we now
search for a transformation that decouples and linearises 
the dynamics. 

To this end, we first investigate the behaviour of the 
$R^{}_G$ and $T^{}_G$ in the discrete-time model, 
since a deeper understanding of their actions will help us find a new transformation.
We are still concerned with the LDE operators from the continuous-time model, 
because of their favourable structure and the existence of the inverse transformation
(M\"{o}bius inversion).
Moreover, as will become clear later, some of them still have the desired features 
and can be adopted directly for the discrete-time model.
First, we need further notation.

\begin{definition}
  Two links $\alpha, \beta \in L$ are called \emph{adjacent} if $\vert \alpha - \beta \vert =1$.
  We say that a subset $\tilde{L} \subseteq L$ is \emph{contiguous} if for any 
  two links $\alpha, \beta \in \tilde{L}$ with $\alpha \leq \beta$, also all links between
  $\alpha$ and $\beta$ belong to $\tilde{L}$ (this includes the case $\tilde{L} = \varnothing$).
  A non-empty contiguous set of links is written as 
  $\tilde{L}=\left\{ \ell^{}_{\mathrm{min}}, \ldots , \ell^{}_{\mathrm{max}}\right\}$.
\end{definition}

Whereas, according to Theorem~\ref{thm:linear}, all recombinators act linearly 
on the solution of the continuous-time recombination equation,
this does not hold for the solution of the discrete-time model in general,
though the following property still holds.

\begin{lemma} \label{lemma:linear}
   Let $\left \{ c^{}_G \mid  G\subseteq L\right\}$ be a family of non-negative numbers with
   $\sum_{ G \subseteq L} c^{}_G = 1$. 
   For an arbitrary $v \in \mathcal{P}(X)$ and for all $K \subseteq L$ 
   with $L \setminus K$ contiguous, one has
     \[
        R^{}_K \Bigl( \sum_{ G \subseteq L} c^{}_G  R^{}_G ( v ) \Bigr)
            \, = \, \sum_{ G \subseteq L} c^{}_G  R^{}_{ G \cup K } ( v )  \, .
     \]
\end{lemma}

\begin{proof}
    When $K = \varnothing$, the claim is clear, because $R^{}_{\varnothing} = \one$
    and $L$ itself is contiguous.
    Otherwise, we have 
    $K = A \cup B$ with $A:=\{\frac{1}{2}, \frac{3}{2}, \ldots, \alpha\}$ 
    and $B:=\{ \beta, \beta +1, \ldots, \frac{2n-1}{2} \}$ for some $ \beta > \alpha$ 
    (this includes the case $K = L$ via $\beta = \alpha + 1$).
    Since we work on $\mathcal{P}(X)$, we have 
    $R^{}_K = R^{}_B R^{}_A$ from Proposition~\ref{recoprop}.
    With the projection $\pi^{}_{i\cdot}:\mathcal{P} (X) \rightarrow \mathcal{P} (X^{}_i) $ 
    onto a single site $i$, we obtain
   \begin{equation}\label{prosingle}
     \pi^{}_{ i \cdot } \Bigl( \sum_{G \subseteq L} c^{}_G  R^{}_G (v) \Bigr)
        =\sum_{ G \subseteq L } c^{}_G  \pi^{}_{ i \cdot } R^{}_G (v)
        =\sum_{ G \subseteq L } c^{}_G  \pi^{}_{ i \cdot } v = \pi^{}_{ i \cdot} v \, ,
   \end{equation}
    since $\pi^{}_{i\cdot}$ is a linear operator and
    $\pi^{}_{i \cdot} R^{}_G (v) = \pi^{}_{i \cdot} v$ by Proposition~\ref{recoprop}.
    For the contiguous set
    $A $ and 
    $w := \sum_{ G \subseteq L}c^{}_G  R^{}_G ( v )$,
    we obtain, with the help of \eqref{prosingle} 
    and a repeated application of Proposition~\ref{recoprop},
   \begin{eqnarray*}
          \lefteqn{ R^{}_A\Bigl( \sum_{G \subseteq L } c^{}_G  R^{}_G (v) \Bigr) = 
               \pi^{}_{ 0 \cdot } w \otimes \cdots \otimes 
               \pi^{}_{ \left \lfloor \alpha \right \rfloor \cdot } w 
                                 \otimes \pi^{}_{> \alpha \cdot} w } \\
           &=& \pi^{}_{ 0\cdot } v \otimes \cdots \otimes 
                  \pi^{}_{ \left \lfloor \alpha \right \rfloor \cdot } v
                                 \otimes \pi^{}_{> \alpha \cdot } w
            =  \sum_{ G \subseteq L } c^{}_G  \bigl( \pi^{}_{ 0 \cdot } v 
                                 \otimes \cdots \otimes 
               \pi^{}_{ \left \lfloor \alpha \right \rfloor \cdot } v \otimes 
               \pi^{}_{ > \alpha \cdot } R^{}_G (v) \bigr) \\
           &=& \sum_{ G \subseteq L } c^{}_G  \bigl( \pi^{}_{ 0 \cdot } R^{}_G (v) 
                                 \otimes \cdots \otimes
               \pi^{}_{ \left \lfloor \alpha \right \rfloor \cdot } R^{}_G (v) \otimes
               \pi^{}_{ > \alpha \cdot } R_G (v) \bigr)
            =  \sum_{ G \subseteq L } c^{}_G  R^{}_{A \cup G } (v) \, .
    \end{eqnarray*}
    An analogous calculation reveals 
    $R^{}_B \bigl( \sum_{G \subseteq L } c^{}_G  R^{}_{A \cup G} (v) \bigr) = 
    \sum_{ G \subseteq L } c^{}_G  R^{}_{A \cup B \cup G } (v)$ 
    for contiguous $B$.
    This proves the assertion. \qed
\end{proof}
The intuitive content of Lemma~\ref{lemma:linear} falls into place with the 
explanation of Theorem~\ref{thm:adevelop}.
The linearity of the particular recombinators of Lemma~\ref{lemma:linear} 
is due to the fact that
$R^{}_K$ produces only one segment,
namely $L\setminus K$, that might be
affected by previous recombination events while all other segments consist of
only one site and thus cannot bring along cuts from `the past'.

\section{Reduction to subsystems}\label{sec:subsys}

In this section, we show a certain product structure of the 
recombinators and the LDE operators. This will turn out as
the key for constructing an appropriate transformation. 
Recall that a crossover at link $\alpha\in L$ 
partitions $S$ into $\left\{0, \ldots, \left\lfloor \alpha\right\rfloor\right\}$ 
and
$\left\{\left\lceil \alpha\right\rceil, \ldots, n \right\}$. 
In general, recombination at the links belonging to 
$G = \bigl\{\alpha^{}_1, \ldots, \alpha^{}_{\left|G\right|}\bigr\}\subseteq L$,
$\alpha^{}_1~< \alpha^{}_2 ~<~ \cdots~ <~ \alpha^{}_{\left|G\right|}$,
induces the following \emph{ordered} partition
$\mathcal{S}^{}_G = \bigl\{J^{G}_0,J^{G}_1, \ldots, J^{G}_{\left|G\right|}\bigr\}$ of $S$ 
(see Fig.~\ref{fig:subsystems}):
\[
    J^{G}_0 = \left\{ 0, \ldots , \left\lfloor \alpha_1 \right\rfloor\right\} ,\;
    J^{G}_1 = \left\{\left\lceil   \alpha_1\right\rceil , \ldots ,
         \left\lfloor \alpha_2\right\rfloor\right\} , \; \ldots ,
    J^{G}_{\left|G\right|} = \bigl\{\lceil  \alpha^{}_{\vert G\vert}\rceil , \ldots , n \bigr\} \, .
\]
Note that the partition is ordered due to the restriction to single crossovers.
In connection with this, we have the sets of links that correspond to the 
respective parts of the partition $\mathcal{S}^{}_G$ (Fig.~\ref{fig:subsystems}).
Namely, for
$G = \bigl\{\alpha^{}_1, \ldots, \alpha^{}_{\left|G\right|}\bigr\}\subseteq L$,
$\mathcal{L}^{}_G := \left\{I_0^G, I_1^G, \ldots, I_{\left|G\right|}^G\right\}$ with
\begin{equation}\label{linksegments}
 \begin{split}
  I_0^G & = \left\{\alpha \in L : \tfrac{1}{2} \leq \alpha < \alpha^{}_1 \right\}, \quad
    I_{\left|G\right|}^G  = 
    \bigl\{\alpha \in L : \alpha^{}_{\vert G\vert} < \alpha \leq \tfrac{2n-1}{2} \bigr\} \, , \\
  \mbox{and }  I_{\ell}^G  &= \left\{\alpha \in L : \alpha^{}_{\ell} < \alpha < \alpha^{}_{\ell+1} \right\}
  \mbox{ for } 1\leq \ell\leq \vert G\vert -1 
 \end{split}
\end{equation}
specifies the links belonging to the respective parts of $\mathcal{S}^{}_G$: 
the links associated with $J^{G}_k~\in~\mathcal{S}^{}_G$, $ 0 \leq k \leq \vert G\vert$, are exactly those of 
$I_k^G~\in~\mathcal{L}^{}_G$ (and vice versa).

\begin{figure}[htb]
\psfrag{0}{$0$}
\psfrag{1}{$1$}
\psfrag{2}{$2$}
\psfrag{3}{$3$}
\psfrag{4}{$4$}
\psfrag{5}{$5$}
\psfrag{6}{$6$}
\psfrag{7}{$7$}
\psfrag{8}{$8$}
\psfrag{9}{$9$}
\psfrag{A}{$\frac{1}{2}$}
\psfrag{B}{$\frac{3}{2}$}
\psfrag{C}{$\frac{5}{2}$}
\psfrag{D}{$\frac{7}{2}$}
\psfrag{E}{$\frac{9}{2}$}
\psfrag{F}{$\frac{11}{2}$}
\psfrag{G}{$\frac{13}{2}$}
\psfrag{H}{$\frac{15}{2}$}
\psfrag{I}{$\frac{17}{2}$}
\psfrag{I0}{$I^{}_0$}
\psfrag{I1}{$I^{}_1$}
\psfrag{I2}{$I^{}_2 = \varnothing$}
\psfrag{I3}{$I^{}_3$}
\psfrag{J0}{$J^{}_0$}
\psfrag{J1}{$J^{}_1$}
\psfrag{J2}{$J^{}_2$}
\psfrag{J3}{$J^{}_3$}

 \includegraphics{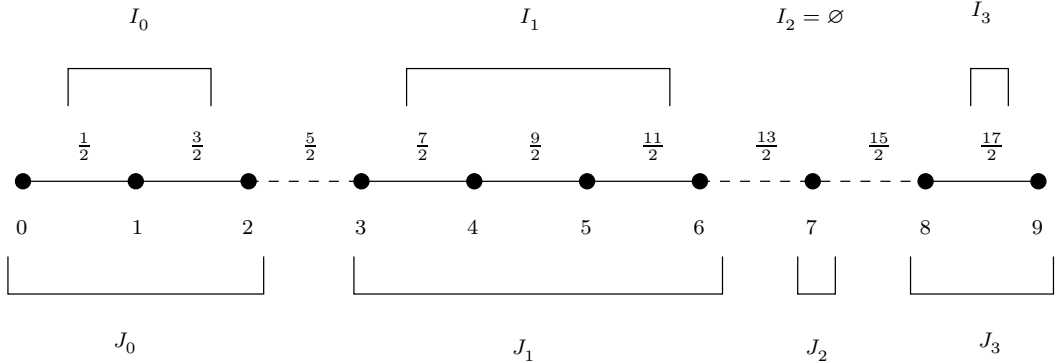}
\caption{A system with $10$ sites (i.e., $S = \{0,\ldots,9 \}$, $L = \{ \frac{1}{2},\ldots,\frac{17}{2} \}$)
cut at the links $G~=~\{\frac{5}{2},\frac{13}{2},\frac{15}{2} \}$ (broken lines).
The resulting subsystems are $\mathcal{S}^{}_G~=~\{J^{}_0,\ldots,J^{}_3 \}$ and
$\mathcal{L}^{}_G = \{I^{}_0,\ldots,I^{}_3\}$ with $J^{}_0 = \{0,1,2\}$, $J^{}_1 = \{ 3,4,5,6\}$,
$J^{}_2 = \{7 \}$ and $J^{}_3 = \{ 8,9 \}$ as well as $I^{}_0 = \{\frac{1}{2},\frac{3}{2} \}$,
$I^{}_1 = \{\frac{7}{2},\frac{9}{2},\frac{11}{2} \}$, $I^{}_2 = \varnothing$ and
$I^{}_3 = \{\frac{17}{2} \}$ (the upper index $G$ is suppressed here for clarity).}

\label{fig:subsystems}
\end{figure}

With this definition, $I_i^G = \varnothing$ is possible for each $0\leq i\leq \vert G\vert$ 
and will be included (possibly multiply) in $\mathcal{L}^{}_G$.
Furthermore, $\mathcal{L}^{}_{\varnothing} := \{ L \}$, so that $I^{\varnothing}_0 = L$.
The upper index will be suppressed in cases where the corresponding set of links is
obvious.
Clearly, $\mathcal{L}^{}_G$ is not a partition of $L$, whereas $\mathcal{S}^{}_G$ is a partition of $S$.

This way, recombination at the links in $G\subseteq L$ produces several `subsystems'
(characterised through the sites $J^{}_k$ and the corresponding links $I^{}_k$,
$0\leq k \leq\vert G\vert$)
with respect to the `full system' described through 
the sites $S$ and the links $L$.
We demonstrate below that it is sufficient to consider these subsystems separately,
a property that reduces the problem of dealing with the recombination dynamics.
Note first that repeated application of \eqref{elemprop1} and \eqref{elemprop2} leads to
\begin{equation}\label{elemprop3}
 \pi^{}_{<\alpha\cdot}R^{(L)}_G(p) = R^{(L^{}_{<\alpha})}_{G^{}_{<\alpha}}(\pi^{}_{<\alpha\cdot}p)
\quad \text{and} \quad 
\pi^{}_{>\alpha\cdot}R^{(L)}_G(p) = R^{(L^{}_{>\alpha})}_{G^{}_{>\alpha}}(\pi^{}_{>\alpha\cdot}p) \, ,
\end{equation}
where 
$R^{(L)}_G$ is our usual recombinator acting on 
$\mathcal{P} (X) = \mathcal{P}(X^{}_0\times \cdots \times X^{}_n)$,
and $ R^{(L^{}_{<\alpha})}_{G^{}_{<\alpha}}$ denotes the respective recombinator on 
$\mathcal{P}(X^{}_{ 0 } \times \cdots \times X^{}_{ \lfloor \alpha \rfloor } )$,
which acts on the subsystem specified through the sites $L^{}_{<\alpha}$
and cuts the links $G^{}_{<\alpha}$
(and analogously for  $ R^{(L^{}_{>\alpha})}_{G^{}_{>\alpha}}$).
Likewise, recombinators $R^{(I^{}_i)}_{H}$, $H\subseteq I^{}_i$,
acting on $\mathcal{P}(X^{}_{J_i})$, may be defined for all subsystems,
$0\leq i\leq \vert G\vert$, in the obvious way.
For consistency, we define $R_{\varnothing}^{(\varnothing)}=\one$.
From now on, the upper index specifies the corresponding system the 
$R^{}_G$ (and, likewise, the $T^{}_G$) are acting on. It will be suppressed in cases where the 
system is obvious.
We now explain the inherent product structure of the recombinators: 

\begin{prop}\label{thm:product}
    Let $G\subseteq L$. For each $\alpha \in G$ and $p \in \mathcal{P}(X)$, one has the identity
    \[
       R_G^{(L)} (p) \, =\, \Bigl( R_{G_{ < \alpha}}^{ (L_{ < \alpha}) } 
                        ( \pi^{}_{ < \alpha \cdot } p ) \Bigr)
       \otimes \Bigl( R_{G_{ > \alpha}}^{(L_{ > \alpha})} 
                        ( \pi^{}_{ > \alpha \cdot } p ) \Bigr) \, .
    \]
\end{prop}

\begin{proof}
    For $\alpha \in G$, Proposition~\ref{recoprop} implies : 
    \begin{eqnarray*}
    R_G^{(L)} (p) &\, =\, & R_{\alpha}^{(L)}\Bigl(R_G^{(L)}(p)\Bigr) 
              \, =\, \Bigl( \pi^{}_{ < \alpha \cdot } R_G^{(L)} (p) \Bigr)
                 \otimes \Bigl(\pi^{}_{ > \alpha \cdot }  R_G^{(L)} (p) \Bigr) \\
              &\,=\,& \Bigl( R_{G_{ < \alpha}}^{( L^{}_{ < \alpha}) } (\pi^{}_{<\alpha\cdot} p ) \Bigr)
       \otimes \Bigl( R_{G^{}_{ > \alpha}}^{(L^{}_{ > \alpha})} ( \pi^{}_{> \alpha \cdot }p )\Bigr) \, ,
    \end{eqnarray*}
    where the last step follows from \eqref{elemprop3}. \qed
 \end{proof}
This proposition carries over to the LDE operators:

\begin{prop}
    On $\mathcal{P} (X) $, the LDE operators satisfy
    \[
       T_G^{(L)} (p) \,=\, \Bigl(T_{ G^{}_{ < \alpha}}^{( L_{ < \alpha}) } ( \pi^{}_{ < \alpha \cdot} p ) \Bigr)
                   \otimes 
                   \Bigl(T_{ G^{}_{ > \alpha}}^{ (L_{ >\alpha} )} ( \pi^{}_{ > \alpha \cdot} p )\Bigr)
                             \quad\text{for all $\alpha\in G$},
     \]
    where $T_{ G^{}_{ < \alpha}}^{( L_{ < \alpha}) }$ and  $T_{ G^{}_{ > \alpha}}^{ (L^{}_{ >\alpha}) }$
    now describe the operators acting on the simplices
    $\mathcal{P} ( X^{}_0 \times \cdots \times X^{}_{ \lfloor \alpha \rfloor})$ and 
    $\mathcal{P} ( X^{}_{ \lceil \alpha \rceil} \times \cdots \times X^{}_n )$, respectively.
\end{prop}

\begin{proof}
    Let $\alpha\in G$. Using the product structure from Proposition~\ref{thm:product} 
    and splitting the sum into two disjoint parts, one obtains
    \begin{eqnarray*}
       \lefteqn{ T_G^{(L)} (p) \,=\, \sum_{ H \supseteq G} (-1)^{ \left|H-G \right|} R^{(L)}_H (p)
              \,=\, \sum_{ H \supseteq G} (-1)^{ \left|H- G\right|} 
                      \biggl(\Bigl( R_{H^{}_{ <\alpha}}^{(L^{}_{ <\alpha})}
                                ( \pi^{}_{<\alpha\cdot}p)\Bigr) \otimes
                   \Bigl( R_{ H^{}_{ > \alpha}}^{( L^{}_{ > \alpha})} 
                             ( \pi^{}_{ > \alpha \cdot } p )\Bigr) \biggr)} \\
              &\,=\,& \!\!\sum_{ L \setminus \left\{ \alpha \right\} \supseteq H \supseteq G
                               \setminus \left \{ \alpha \right\}} \!\!
                          (-1)^{ \left| H-G \setminus \left \{\alpha\right\}\right|}
            \biggl(\Bigl( R_{ H^{}_{<\alpha}}^{( L^{}_{<\alpha})} ( \pi^{}_{ < \alpha \cdot } p )\Bigr)
                      \otimes \Bigl( R_{ H^{}_{ > \alpha}}^{( L^{}_{ > \alpha})}
                              ( \pi^{}_{ > \alpha \cdot} p) \Bigr)\biggl) \\
             &\,=\,&  \!\!\!\! \sum_{ L^{}_ {< \alpha} \supseteq H^{}_{ < \alpha} \supseteq G^{}_{<\alpha}} \!\!\!\!
                     (-1)^{  \left| H_{ < \alpha}-G^{}_{ < \alpha} \right|} \!\!\!\!
                     \sum_{  L^{}_{ > \alpha} \supseteq H^{}_{ > \alpha} \supseteq G^{}_{ > \alpha}} \!\!\!\!
                     (-1)^{ \left| H^{}_{ > \alpha} - G^{}_{ > \alpha} \right|}
                     \biggl( R_{ H^{}_{ < \alpha}}^{ (L^{}_{ < \alpha})}
                           ({ \pi^{}_{ < \alpha \cdot } p} )
                     \otimes  R_{  H^{}_{ > \alpha}}^{ ( L^{}_{ > \alpha})}
                           ({ \pi^{}_{ > \alpha \cdot }p }) \biggr)  \\
                  &\,=\,&\Bigl( \sum_{ L^{}_{ < \alpha} \supseteq H^{}_{ < \alpha } \supseteq G^{}_{ < \alpha}}
                           (-1)^{ \left| H^{}_{ < \alpha }-G^{}_{ < \alpha} \right|}
                     \Bigl( R_{ H^{}_{ < \alpha}}^{ (L^{}_{ < \alpha})} ( \pi^{}_{ < \alpha \cdot} p) \Bigr) \Bigr)  \\
                     &\phantom{=}& \otimes \Bigl( \sum_{ L^{}_{ > \alpha } \supseteq H^{}_{ > \alpha } \supseteq G^{}_{ > \alpha}}
                           (-1)^{ \left| H^{}_{ > \alpha } - G^{}_{ > \alpha} \right|}
                      \Bigl( R_{ H^{}_{ > \alpha}}^{( L^{}_{ > \alpha})} (\pi^{}_{>\alpha\cdot}p)\Bigr)\Bigr)\\
                  &\,=\,& \Bigl(T_{ G^{}_{ < \alpha}}^{( L^{}_{ < \alpha})} ( \pi^{}_{ < \alpha \cdot} p )\Bigr)
                      \otimes \Bigl( T_{ G^{}_{ > \alpha}}^{( L^{}_{ > \alpha})} ( \pi^{}_{ > \alpha \cdot} p )\Bigr) \, ,
                    \end{eqnarray*}
    which establishes the claim. \qed
\end{proof}
Using this argument iteratively on the respective segments, one easily obtains
\begin{equation} \label{RundT}
    T_G^{(L)} (p) \,=\, \Bigl(T_{ \varnothing}^{( I^{}_0)} ( \pi^{}_{ J^{}_0 \cdot}p )\Bigr) \otimes 
                \Bigl(T_{ \varnothing}^{( I^{}_1)} ( \pi^{}_{ J^{}_1 \cdot}p )\Bigr) \otimes\cdots\otimes 
               \Bigl( T_{ \varnothing}^{( I_{\vert G\vert})}(\pi^{}_{J_{ \vert G \vert} \cdot } p )\Bigr)\, ,
\end{equation}
where the upper index specifies the corresponding subsystems associated with $G$,
compare~\eqref{linksegments}.
Hence, the effect of $T_G^{}$ on the full system is given by that of
$T_{\varnothing}^{}$ on the respective subsystems corresponding to $G$.

Our goal is now to study the effect of the $R^{}_G$ and $T^{}_G$ on $\varPhi$,
the right-hand side of the recombination equation~\eqref{solutiondiscrete}. 
This will show us in more detail when and why the LDE operators
from the continuous-time model are not sufficient for solving the discrete-time model
and, at the same time, will direct us to the new transformation.

If $\varPhi^{ (L)}$ 
denotes the right-hand side of the recombination equation on the full simplex
$\mathcal{P}( X^{}_0 \times \cdots \times X^{}_n)$,
then, for any contiguous $I = \left\{ \alpha, \ldots, \beta \right\} \subseteq L$,
the right-hand side of the recombination equation on the subsimplex
$\mathcal{P} (X^{}_{ \lfloor \alpha \rfloor} \times \cdots \times X^{}_{ \lceil \beta \rceil})$ 
will be denoted with $\varPhi^{ ( I )}$.
Again, we suppress the upper index when the simplex is obvious.

\begin{prop}
    For the right-hand side of the recombination equation,
    \[ \varPhi^{ ( L )} (p) \,=\, \eta \, p +
       \sum_{ \alpha \in L } \rho^{}_{ \alpha} R_{ \alpha }^{(L)} (p) 
            \,=\, p + \sum_{ \alpha \in L } \rho^{}_{ \alpha } 
                 \bigl( R_{\alpha}^{( L)}-\one^{( L)} \bigr) (p)  \, ,
    \]
    one finds
    \[    
      R^{(L)}_{ \alpha} \bigl( \varPhi^{ ( L ) } (p) \bigr)\, =\,
              \Bigl(\varPhi^{ ( L_{ < \alpha} )} ( \pi^{}_{ < \alpha \cdot } p)\Bigr) \otimes
              \Bigl(\varPhi^{ ( L_{ > \alpha} )} ( \pi^{}_{ > \alpha \cdot} p )\Bigr)
    \]
    for every $\alpha\in L$ and $p\in \mathcal{P}(X)$.
\end{prop}

\begin{proof}
    Since
       $R^{(L)}_{ \alpha } \bigl( \varPhi^{ ( L ) } ( p ) \bigr) \,=\,
      \Bigl(\pi^{}_{ < \alpha \cdot } \bigl( \varPhi^{ ( L ) } ( p ) \bigr)\Bigr) \otimes
      \Bigl(\pi^{}_{ > \alpha \cdot } \bigl( \varPhi^{ ( L ) } ( p ) \bigr)\Bigr)$,
    we obtain with the help of~\eqref{elemprop3}: 
    \begin{eqnarray*}
       \lefteqn{ \pi^{}_{ < \alpha \cdot } \bigl( \varPhi^{ ( L ) } ( p ) \bigr) \,=\,
       \pi^{}_{ < \alpha \cdot } 
         \Bigl(p + \sum_{ \beta \in L } \rho^{}_{ \beta } \bigl( R_{ \beta }^{ ( L)}
                                   -\one^{(L)}\bigr) (p) \Bigr)} \\
       &\,=\,&  \pi^{}_{ < \alpha \cdot } p + \sum_{ \beta < \alpha} \rho^{}_{ \beta } 
         ( R_{ \beta }^{( L^{}_{ < \alpha} )} - \one^{(L^{}_{ < \alpha }) }) 
            (\pi^{}_{ < \alpha \cdot } p) +  
         \sum_{ \beta \geq \alpha} \rho^{}_{ \beta } 
         ( R_{ \varnothing }^{( L^{}_{ < \alpha}) } - \one^{(L^{}_{ < \alpha }) }) 
            (\pi^{}_{ < \alpha \cdot } p)  \\
       &\,=\,& \pi^{}_{ < \alpha \cdot } p + \sum_{\beta < \alpha} \rho^{}_{ \beta } 
          ( R_{ \beta }^{ ( L^{}_{ < \alpha} )} - \one^{(L^{}_{ < \alpha }) }) 
            (\pi^{}_{ < \alpha \cdot } p) 
       \, =\,  \varPhi^{ ( L^{}_{ < \alpha} )} ( \pi^{}_{ < \alpha \cdot } p) \, .
    \end{eqnarray*}
    Analogously, one obtains 
    $\pi^{}_{ > \alpha \cdot } \left( \varPhi^{ ( L ) } ( p ) \right)
     = \varPhi^{ ( L^{}_{ > \alpha} )} ( \pi^{}_{ > \alpha \cdot} p )$, 
     and the assertion follows. \qed
\end{proof}
More generally, this theorem implies inductively that
\begin{equation}\label{recomonphi}
   R^{(L)}_G ( \varPhi^{ (L) } (p)) \, = \,
                \Bigl(\varPhi^{ ( I^{}_0 ) } ( \pi^{}_{ J^{}_0 \cdot} p )\Bigr) \otimes
                \Bigl(\varPhi^{ ( I^{}_1 ) } ( \pi^{}_{ J^{}_1 \cdot} p )\Bigr) \otimes \cdots \otimes
                \Bigl(\varPhi^{ ( I^{}_{\vert G\vert} )} (\pi^{}_{J^{}_{ \vert G \vert}\cdot} p )\Bigr) \, .
\end{equation}
Finally, for the interaction between the $T^{(L)}_G$ and $\varPhi^{ (L)}$,
we have the following result.

\begin{prop}\label{thm:prodLDEops}
    For the LDE operators~\eqref{trafo} and all $G\subseteq L$,
    one has
    \[
         T_G^{(L)} \bigl( \varPhi^{ ( L) } (p) \bigr) \, =\,
         \Bigl(T_{ \varnothing}^{( I^{}_0)} \bigl( \varPhi^{ ( I^{}_0 )} ( \pi^{}_{J^{}_0 \cdot} p) \bigr)\Bigr)
         \otimes\Bigl( T_{ \varnothing}^{( I^{}_1)}\bigl( \varPhi^{( I^{}_1)} ( \pi^{}_{J^{}_1 \cdot} p ) \bigr)\Bigr)
         \otimes \cdots \otimes 
         \Bigl(T_{ \varnothing}^{( I^{}_{\vert G\vert})}
         \bigl( \varPhi^{( I^{}_{\vert G\vert} ) } ( \pi^{}_{J^{}_{\vert G\vert} \cdot} p ) \bigr)\Bigr) \, ,
       \]
    with $I^{}_0,\ldots, I^{}_{\vert G\vert}$ according to \eqref{linksegments}.
\end{prop}

\begin{proof}
    Using \eqref{recomonphi} and \eqref{RundT}, one calculates 
     \begin{eqnarray*}
        T_G^{(L)} \bigl( \varPhi^{ ( L ) } (p) \bigr)
          &\,=\,& T_G^{(L)} \Bigl( R^{(L)}_G \bigl( \varPhi^{ ( L )} (p) \bigr) \Bigr)  
          \,=\, T_G^{(L)}\Bigl(
              \bigl( \varPhi^{ ( I^{}_0 )} (\pi^{}_{ J^{}_0 \cdot }p )\bigr) \otimes \cdots \otimes
              \bigl(\varPhi^{ ( I^{}_{ \vert G \vert} )}(\pi^{}_{J_{ \vert G \vert} \cdot}p)\bigr)\Bigr)  \\
          &\,=\,& \Bigl(T_{ \varnothing }^{( I^{}_0)} \bigl( \varPhi^{ ( I^{}_0 )} ( \pi^{}_{J^{}_0\cdot} p )\bigr )\Bigr)
              \otimes \cdots \otimes \Bigl(T_{ \varnothing}^{(I^{}_{ \vert G \vert})}
              \bigl( \varPhi^{ ( I^{}_{ \vert G \vert} ) } 
               ( \pi^{}_{J^{}_{ \vert G \vert } \cdot} p) \bigr)\Bigr) \, ,
\end{eqnarray*}
    which establishes the formula. \qed
\end{proof}
This result is of particular significance since it shows that,
to determine the effect of the $T^{(L)}_G$  on $\varPhi$,
it is sufficient to know the action of the 
$T^{}_{\varnothing}$ on the subsystems that correspond to $G$.
Hence, we  now  need to determine  $T^{}_{ \varnothing} \circ \varPhi$.
It will turn out that this relies crucially on the commutators of $R^{}_G$ with $\varPhi$,
which will be the subject of the next section.

\section{The commutator and linearisation}\label{sec:comm}
The more algebraic approach of \cite{MB}, which was later 
generalised by Popa~\cite{Popa}, suggests to further analyse the problem in terms of
commuting versus non-commuting quantities.
For $G\subseteq L$, the commutator is defined as 
$\left[ \, R^{}_G, \varPhi \, \right ] := R^{}_G \circ \varPhi - \varPhi \circ R^{}_G$.
Recall that, in the continuous-time model, the linear action of the recombinators
on the solution of the differential equations entails that the corresponding
forward flow commutes with each recombinator (see Corollary~\ref{coro1}). 
But this no longer holds for discrete time:
$\left[ \, R^{}_G, \varPhi \, \right ] = 0$ is not true in general. 
We are interested in the commutators because --- as we will see  in a moment --- 
they lead us to the evaluation of 
$T^{}_{ \varnothing} \circ \varPhi$, and this in turn gives $T^{}_G \circ \varPhi$
(see Proposition~\ref{thm:prodLDEops}).

\begin{prop}\label{prop:tempty}
    Let $\eta = 1- \sum_{ \alpha \in L}\rho_{ \alpha}$ as before. 
    On $\mathcal{P}(X)$, one has
    \[
       T^{}_{ \varnothing} \circ \varPhi \,=\, \eta \, T^{}_{ \varnothing } +
           \sum_{ G \subseteq L}(-1)^{ \left| G \right|} \left[ \, R^{}_G,\varPhi \, \right].
    \]
\end{prop}

\begin{proof}
     Expressing the left-hand side as
    \[
       T^{}_{ \varnothing} \circ \varPhi \,=\, \sum_{ G \subseteq L} 
         (-1)^{ \left| G\right|} ( R^{}_G \circ \varPhi ) \,=\, 
         \sum_{ G \subseteq L} (-1)^{ \left| G \right|} ( \varPhi \circ R^{}_G )
        +\sum_{ G\subseteq L} (-1)^{ \left| G \right|} \left[ \, R^{}_G, \varPhi \, \right] \, ,
    \]
    and using  $\varPhi^{ } = \eta \thinspace\one +
       \sum_{ \alpha \in L } \rho^{}_{ \alpha} R_{ \alpha }^{}$, one calculates
\begin{eqnarray*}
    \sum_{ G \subseteq L} (-1)^{ \left| G \right|} ( \varPhi \circ R^{}_G )&\,=\,&
    \sum_{ \alpha \in L} \Bigl( \sum_{ G \subseteq L} (-1)^{ \left| G \right|}
    \rho^{}_{ \alpha} R^{}_{ \alpha } R^{}_G \Bigr)
    +\eta \sum_{ G \subseteq L} (-1)^{ \left| G \right|} R^{}_G \\[2mm]
   &\,=\,& \eta \, T^{}_{ \varnothing} + 
     \sum_{ \alpha \in L} \sum_{ \substack{ G \subseteq L\\\alpha\notin G}}
     \left( (-1)^{ \left| G \right|} \rho^{}_{ \alpha} R^{}_{ \alpha} R^{}_G +
     (-1)^{\left|G \cup \left\{ \alpha \right\} \right|}
      \rho^{}_{ \alpha} R^{}_{ G \cup \left\{ \alpha \right\}} \right)  \\
   &\,=\,& \eta \, T^{}_{ \varnothing} + \sum_{\alpha\in L}
       \Bigl( \sum_{ \substack{G \subseteq L\\\alpha\notin G}}
     (-1)^{ \left| G \right|} \rho^{}_{ \alpha}
     ( R^{}_{ \alpha} R^{}_G - R^{}_{ G \cup \left\{\alpha\right\}}) \Bigr) \,=\, 
     \eta T^{}_{ \varnothing} \, ,
\end{eqnarray*}
    which shows the claim. \qed
\end{proof}
Proposition~\ref{prop:tempty} shows that 
$T^{}_{\varnothing}$ only yields a diagonal component
if \emph{all} recombinators commute with $\varPhi$.
We now need to determine the commutator $\left[ \, R^{}_G,\varPhi \, \right]$.
To this end, it is advantageous to introduce a new set of operators.
\begin{definition} \label{defin:newop}
For $G \subseteq K \subseteq L$, we define the operators
    \begin{equation}\label{newoperators}
          \widetilde{T}^{}_{G,K} \,:= \sum_{ G \subseteq H \subseteq K} 
                                  (-1)^{ \left| H-G \right| } R^{}_H \, .
    \end{equation}
Equivalently, for any $M\subseteq L\setminus G$, this means that
\[
    \widetilde{T}^{}_{ G, G \dot{\cup} M} \, = \sum_{G\subseteq H\subseteq G \dot{\cup} M}
                                      (-1)^{ \left| H-G \right| } R^{}_H
        = \sum_{ K \subseteq M}  (-1)^{ \left| K \right| } R^{}_{G \dot{\cup} K} \, .
\]
\end{definition}
These operators act on the full simplex and can be interpreted in analogy 
to the original LDE operators \eqref{trafo}, where the links in the complement of
$G \dot{\cup} M$ (the disjoint union of $G$ and $M$) are regarded as inseparable.
If necessary, we will specify the system the operators are acting on 
by an upper index as before. 

\begin{lemma}
    On $\mathcal{P}(X)$, the operators~\eqref{newoperators} satisfy
    \[
      \widetilde{T}^{}_{ G, G }\, =\, R^{}_G \quad \text{and} 
      \quad \widetilde{T}^{}_{ G, L } \,=\, T^{}_G \, .
     \]
   They have a product structure, 
 \begin{equation}\label{prodtilde}
         \widetilde{T}_{G,G\dot{\cup} H}^{(L)} (p) \,=\,
         \Bigl(\widetilde{T}_{ \varnothing, H\cap I^{G}_0}^{( I^{G}_0)}  ( \pi^{}_{J^{G}_0 \cdot} p)\Bigr) 
         \otimes \Bigl(\widetilde{T}_{ \varnothing, H \cap I^{G}_1}^{ (I^{G}_1)}( \pi^{}_{J^{G}_1 \cdot} p )\Bigr) 
         \otimes \cdots \otimes 
         \Bigl(\widetilde{T}_{ \varnothing, H \cap I^{G}_{\vert G\vert}}^{ (I^{G}_{\vert G\vert})} 
          ( \pi^{}_{J^{G}_{\vert G\vert} \cdot} p )\Bigr) \, ,
       \end{equation}
 for all $H\subseteq L\setminus G$.
    Moreover, one has 
           \begin{equation}\label{newTtilde1}
            \widetilde{T}^{}_{ G, G \dot{\cup} M}\, = 
                 \sum_{ G \subseteq H \subseteq L \setminus M } T^{}_H \,
                  = \sum_{K\subseteq L\setminus{(M \cup G)}} T^{}_{G \dot{\cup} K} 
           \end{equation}
     for all $G, M\subseteq L$ with $G\cap M = \varnothing$.
   Consequently, M\"{o}bius inversion returns $T^{}_G$ as
    \begin{equation} \label{newTtilde2}
                 T^{}_G \,= \sum_{G \subseteq H \subseteq L \setminus M }
                 (-1)^{ \vert H-G\vert} \, \widetilde{T}^{}_{ H, H \dot{\cup} M} \, .
            \end{equation}
\end{lemma}

\begin{proof}
The first assertion is obvious;
the second is analogous to \eqref{RundT} and follows along the same lines.
Relation~\eqref{newTtilde1} is true since
\[
 \begin{split}
  \widetilde{T}^{}_{ G, G \dot{\cup} M} \,&= \sum_{K\subseteq M} (-1)^{\vert K\vert} R^{}_{G \dot{\cup} K}
     \,= \sum_{K\subseteq M} (-1)^{\vert K\vert} \sum_{H\supseteq G \dot{\cup} K} T^{}_H \\[2mm]
     &= \sum_{H\supseteq G} \sum_{\substack{K \subseteq M \\ K \subseteq H}} (-1)^{\vert K\vert}\,  T^{}_H
     \,= \sum_{H\supseteq G} T^{}_H \sum_{K\subseteq M \cap H} (-1)^{\vert K\vert} \\
     &= \sum_{H\supseteq G} \, \delta_{M \cap H , \varnothing} \, T^{}_H 
      = \sum_{G \subseteq H \subseteq L\setminus M} T^{}_H \, .
  \end{split}
\]
In the second-last step, we used that, if $H$ is a finite set, one has 
\begin{equation}\label{binlehr}
 \sum_{G\subseteq H} (-1)^{\vert G \vert} = \delta^{}_{H , \varnothing} \, ,
\end{equation}
which is the
key property of the M\"{o}bius function of ordered partitions. \qed
\end{proof}

Before we turn to the commutator, we introduce a new function,
the  \emph{separation function},
which will allow for a clear and compact notation.

\begin{definition}
For $G, H \subseteq L$ with $G \cap H = \varnothing$, we say that $G$ \emph{separates} $H$
if, 
for all $\alpha, \beta \in H$ with $\alpha < \beta$, there is a $\gamma \in G$ with
$\alpha < \gamma < \beta$.
Hence, we define the separation function as
\begin{equation*}
   \sep (G,H) = 
    \begin{cases}
     1  \text{, if $G$ separates $H$,}\\
     0  \text{, otherwise.}
     \end{cases}
\end{equation*}
In the particular cases $H = \varnothing$ and $H = \{ \alpha \}$, $\alpha \in L$,  we define
$\sep(G,H) = 1$ for all $G\subseteq L$,
and it is understood that $\sep(G,H) = 0$ whenever $G \cap H \neq \varnothing$.
\end{definition}
First, let us summarise some elementary properties
of the separation function.
\begin{lemma} \label{lemmasepfunc}
The separation function $\sep (G,H)$ with $H \subseteq L\setminus G$ has the 
following properties:
\begin{enumerate}\rm
 \item $\sep(G,H) = 0$, if $H$ contains any adjacent links;

 \item  $\sep(G,H) = 0 \ \text{implies}\ \sep(G^{'},H) = 0$ for all  
       $ G^{'} \subseteq G$;

\item  $\sep(G,H) = 0$  whenever $L\setminus G$ is contiguous with 
$H\subseteq L\setminus G$ and $\vert H\vert\geq 2$;
\item  $\sep(G,H) = 1 \ \text{implies}\ I_i^{H} \cap G \neq \varnothing$
        for all $i \in \{ 1, \ldots,\vert H\vert -1\}$. \qed
\end{enumerate}
\end{lemma}
Later, we need the following summation formula for the separation function.
\begin{lemma}\label{lemmamoresepfunc}
 Let $H,K\subseteq L$ with $H\neq\varnothing$, $H\cap K = \varnothing$,
and $I_i^{H}$ defined as in \eqref{linksegments}. Then
\[
  \sum_{G\subseteq K} (-1)^{\vert G\vert} \sep(G,H) \,=\,
      \sep(K,H)\thinspace (-1)^{\vert H\vert -1} \, \delta_{K\cap I_0^{H}, \varnothing}
          \, \delta_{K\cap I_{\vert H\vert}^{H}, \varnothing} \, .
\]
\end{lemma}
\begin{proof}
 For $\sep(K,H) = 0$, the claim is clear by Lemma~\ref{lemmasepfunc}($2$).
We now define $A^{}_i := K \cap I^{H}_{i}$ for all $i \in \{0,\ldots, \vert H \vert \}$.
Then, for $\sep(K,H) =1$, it follows from Lemma~\ref{lemmasepfunc}($4$) that
$A^{}_j \neq \varnothing$ for all $1\leq j \leq \vert H \vert -1$.
Likewise, since $G\subseteq K$, $\sep(G,H)=1$ if and only if 
$G\cap I_j^{H} \neq \varnothing$ for all $1\leq j \leq \vert H \vert -1$,
with no condition emerging for $G\cap I_0^{H}$ or $G\cap I_{\vert H\vert}^{H}$.
This gives
\[
   \sum_{G\subseteq K} (-1)^{\vert G\vert} \sep(G,H) \,= 
    \sum_{B^{}_0 \subseteq A^{}_0} (-1)^{\vert B^{}_0 \vert}
   \prod_{i=1}^{\vert H\vert - 1} 
   \Bigl(\sum_{\substack{ B^{}_i \subseteq A^{}_i \\ B^{}_i \neq \varnothing} } (-1)^{\vert B^{}_i \vert} \Bigr)
   \sum_{B^{}_{\vert H\vert} \subseteq A^{}_{\vert H\vert}} (-1)^{\vert B^{}_{\vert H\vert} \vert}
   \,= \prod_{j=0}^{\vert H\vert} F_j \, .
\]
Here, for $j= 0$ and $j = \vert H\vert$, the factors $F^{}_j$ are given by 
$F_j :=\sum_{B^{}_j \subseteq A^{}_j} (-1)^{\vert B^{}_j \vert} = \delta_{A_j,\varnothing}$,
where we have used \eqref{binlehr}.
For $1\leq j \leq \vert H \vert -1$,
\[
 F_j \,:= \, \sum_{\substack{ B^{}_j \subseteq A^{}_j \\ B^{}_j \neq \varnothing} }
         (-1)^{\vert B^{}_j\vert} = -1 + \sum_{B^{}_j\subseteq A^{}_j} (-1)^{\vert B^{}_j\vert}
      = -1 + \delta_{A^{}_j,\varnothing} = -1 \, ,
\]
where we have again used \eqref{binlehr} in the second-last step,
and $A^{}_j \neq \varnothing$ in the last. \qed
\end{proof}

With this notation, let us take a closer look at $R_G^{(L)}\bigl( \varPhi^{(L)}(p)\bigr)$
for $G\subseteq L$.
Evaluating \eqref{recomonphi} explicitly,
using Definition~\ref{defin:newop}, expanding and using the 
product structure~\eqref{prodtilde} backwards gives
\begin{equation*}
\begin{split}
    R_G^{(L)} \bigl( \varPhi^{(L)}(p) \bigr) \,&=\,
    \Bigl( \pi^{}_{ J^{}_0 \cdot }p + \sum_{\alpha^{}_0\in I^{}_0}\rho^{}_{\alpha^{}_0}
     (R_{ \alpha^{}_0}^{( I^{}_0 )} - \one^{(I^{}_0)}) ( \pi^{}_{J^{}_0 \cdot} p )
      \Bigr) \otimes \cdots \otimes  \\
      &\phantom{=} \Bigl( \pi^{}_{ J^{}_{ \vert G \vert} \cdot} p +
         \sum_{\alpha^{}_{ \vert G \vert} \in I^{}_{\vert G\vert}}
      \rho^{}_{ \alpha^{}_{ \vert G \vert}} ( R_{\alpha^{}_{ \vert G \vert }}^{(I^{}_{\vert G\vert})}
      - \one^{( I^{}_{ \vert G \vert})}) ( \pi^{}_{ J^{}_{ \vert G \vert} \cdot} p )
        \Bigr)  \\
  &= \,\Bigl(\one^{(I^{}_0)} - \sum_{\alpha^{}_0 \in I^{}_0} \rho^{}_{\alpha^{}_0} 
        \widetilde{T}_{ \varnothing, \alpha^{}_0}^{( I^{}_0)}\Bigr)  ( \pi^{}_{J^{}_0 \cdot} p)
    \otimes\cdots\otimes \Bigl(\one^{( I^{}_{ \vert G \vert})} -
         \sum_{\alpha^{}_{\vert G\vert} \in I^{}_{\vert G\vert}} \rho^{}_{\alpha^{}_{\vert G\vert}} 
        \widetilde{T}_{ \varnothing, \alpha^{}_{\vert G\vert}}^{( I^{}_0 ) } \Bigr) 
                ( \pi^{}_{J^{}_{\vert G\vert} \cdot} p)  \\
  &= \sum_{H\subseteq L\setminus G} (-1)^{\vert H\vert}\sep(G,H)\rho^{}_H  \, \widetilde{T}^{(L)}_{G,G\dot{\cup} H}(p) \, ,
\end{split}
\end{equation*}
where, in the last step, we have further set
$\rho^{}_H = \prod_{\alpha \in H} \rho^{}_{\alpha}$ for all $H\subseteq L$
(in particular, $\rho^{}_{\varnothing} = 1$) and used Lemma~\ref{lemmasepfunc}($4$);
note that the separation function is basically used as an indicator variable here.
On the other hand, we obtain
\begin{eqnarray*}
   \varPhi \circ R^{}_G  \,&=& \,\sum_{ \alpha \in L \setminus G} 
            \rho^{}_{ \alpha } R^{}_{\alpha} R^{}_G 
            + \bigl( 1 - \sum_{\alpha\in L\setminus G} 
             \rho^{}_{ \alpha} \bigr) R^{}_G  \\
          &=& \,\widetilde{T}^{}_{ G , G }  -
                \sum_{\alpha \in L\setminus G}
                      \rho^{}_{\alpha}
                      \widetilde{T}^{}_{ G , G \dot{\cup} \{ \alpha \}} \\
          &=& \,\sep(G,\varnothing) \, \widetilde{T}^{}_{ G , G }
              - \sum_{\alpha \in L\setminus G} \sep(G, \{\alpha\}) 
                 \rho^{}_{\alpha} \widetilde{T}^{}_{ G , G \dot{\cup} \{ \alpha \}} \, ,
\end{eqnarray*}
which finally yields the commutator. 

\begin{theorem}\label{commutator}
For all $G\subseteq L$,  the commutator on $\mathcal{P}(X)$ is given by
\begin{equation*} 
    \left[ \, R^{}_G , \varPhi \, \right]  \,= 
       \sum_{\substack{ H \subseteq L \setminus G \\ \vert H \vert \geq 2} } 
       (-1)^{\vert H \vert} \sep(G,H) \,
       \rho^{}_H \ts \widetilde{T}^{}_{ G, G \dot{\cup} H} \, .  
\end{equation*} 
\qed
\end{theorem}
Please note that, by the properties of the separation function, many of the summands vanish.
In particular, $\left[ \, R^{}_G , \varPhi \, \right] = 0$ whenever
$\vert L\setminus G \vert \leq 1$.

\begin{coro} \label{thm:commute}
    $\left[ \, R^{}_G , \varPhi \, \right] = 0$ if 
    $L \setminus G$ is contiguous. 
\end{coro}

\begin{proof}
By Theorem~\ref{commutator}, only terms with $\vert H\vert \geq 2$
need be considered. For these, Lemma~\ref{lemmasepfunc}($3$) tells us that
$\sep(G,H) = 0$ if $L\setminus G$ is contiguous and $H\subseteq L\setminus G$.
Hence,  $\left[ \, R^{}_G , \varPhi \, \right] = 0$.
\qed
\end{proof}
Let us note in passing that the converse direction of Corollary~\ref{thm:commute}
may fail if the site spaces are sufficiently trivial.
Nevertheless, in the generic case, $\left[ \, R^{}_G , \varPhi \, \right] = 0$ 
implies $\sep(G,H) = 0$ for all $H\subseteq L\setminus G$ with $\vert H\vert \geq 2$,
because the relevant terms then cannot cancel each other.
We omit a more precise discussion of this point, because we do not need it later on.

Recalling that $\varPhi^{t}$ is the discrete-time analogue of $\varphi^{}_t$,
we can consider Corollary~\ref{thm:commute} 
as what is left of Corollary~\ref{coro1} in discrete time.
Hence, it becomes clear why the LDE operators~\eqref{trafo} from the continuous-time
model do not suffice to linearise \emph{and} decouple the discrete-time dynamics.

We still aim at determining $T^{}_{\varnothing} \circ \varPhi$
according to Proposition~\ref{prop:tempty}, expressing the commutator 
 $\left[ \, R^{}_G , \varPhi \, \right]$ in terms of the $T^{}_G$ (which
are related to the  $\widetilde{T}^{}_{ G, G \dot{\cup} M}$ via~\eqref{newTtilde1}).

\begin{theorem} \label{thm:tvarnotcomplete}
   On $\mathcal{P}(X)$, the operators $T^{}_{G} = T^{(L)}_G$
   and  $\varPhi = \varPhi^{(L)}$  satisfy
\begin{equation}\label{TGonPhi}
    T_G^{(L)} \circ \varPhi^{ (L)} \,
         = \sum_{ K \supseteq G }z^{ ( L)} (G,K)\: T_K^{(L)} 
\end{equation}
for all $G\subseteq L$. The coefficients $z^{ ( L) } (\varnothing, K )$, $K\subseteq L$,
are given by 
\begin{equation}\label{zleerleer}
 z^{(L)} ( \varnothing ,\varnothing)\, =\, 1- \sum_{\alpha \in L}\rho^{}_{\alpha}
\end{equation}
and, for $K \neq \varnothing$, by 
\begin{equation}\label{zleerk}
 z^{(L)} ( \varnothing , K) \,=\, 
 -\sum_{H\subseteq L\setminus K}\rho^{}_{H} \sep(K,H)\,
      (1- \delta^{}_{H\cap I_0^{K},\varnothing})\, (1- \delta^{}_{H\cap I^{K}_{\vert K \vert},\varnothing}) \, .
\end{equation}
For $K\supseteq G \neq \varnothing$, the coefficients are recursively determined by
\[ z^{ ( L ) } (G,K)
   \,=\, z^{ ( I_0^G ) } \left( \varnothing,  K\cap I_0^G \right)
     \cdot \ldots \cdot z^{ ( I_{ \vert G \vert}^G )} 
     \bigl( \varnothing, K\cap I_{ \vert G\vert}^G  \bigr) \, .\]
\end{theorem}

\begin{proof}
    Let us first prove the case $G = \varnothing$.
According to Proposition~\ref{prop:tempty}, we have
     $ T^{}_{ \varnothing} \circ \varPhi =
        \eta T^{}_{ \varnothing} + \sum_{ G^{'} \subseteq L} (-1)^{ \vert G^{'} \vert} [ \, R^{}_{G^{'}} , \varPhi \, ]$,
    where $\eta = z^{(L)}(\varnothing, \varnothing)$ by definition.
    Let us thus evaluate the last term. In the first step,
   we insert the commutator from Theorem~\ref{commutator};
    we then use Definition~\ref{defin:newop} and
    change the order of summation to arrive at
    \begin{eqnarray}\label{toTGphi}
      \sum_{ G^{'} \subseteq L} (-1)^{ \vert G^{'} \vert} [ \, R^{}_{G^{'}} , \varPhi \, ] \,&=&\,
      \sum_{ G^{'} \subseteq L} (-1)^{ \vert G^{'} \vert} 
           \sum_{\substack{ H \subseteq L \setminus G^{'} \\ \vert H \vert \geq 2} } 
       (-1)^{\vert H \vert} \sep(G^{'},H) \ts
       \rho^{}_H \ts \widetilde{T}^{}_{ G^{'}, G^{'} \dot{\cup} H} \nonumber \\
      &=& \sum_{ G^{'} \subseteq L} (-1)^{ \vert G^{'} \vert} 
           \sum_{\substack{ H \subseteq L \setminus G^{'} \\ \vert H \vert \geq 2} } 
       (-1)^{\vert H \vert} \sep(G^{'},H) \ts
       \rho^{}_H \sum_{G^{'} \subseteq K \subseteq L\setminus H} T^{}_K \nonumber \\
      &=& \sum_{K\subseteq L} T_K^{}\sum_{\substack{ H \subseteq L\setminus K \\ \vert H \vert \geq 2} }
             (-1)^{\vert H \vert} \rho^{}_H 
             \sum_{ G^{'}\subseteq K} (-1)^{\vert G^{'}\vert}\sep(G^{'},H) \, ,
    \end{eqnarray}
which does not contain any term with $T^{}_{\varnothing}$.
We can now compare coefficients for $T^{}_K$.
Note first that, by \eqref{toTGphi}, we only need to consider sets $H\subseteq L\setminus K$,
that is, $H\cap K = \varnothing$. In this case,
 $\delta^{}_{K\cap I_0^{H},\varnothing} = 1 - \delta^{}_{H\cap I_0^{K},\varnothing}$ and
$\delta^{}_{K \cap I_{\vert H\vert}^{H},\varnothing} = 1 - \delta^{}_{H\cap I_{\vert K\vert}^{K},\varnothing}$.
This is true since $K\cap I_0^{H} = \varnothing \; (\neq \varnothing)$ implies that the smallest element in $H$
is smaller (larger) than the smallest element in $K$, thus $H\cap I_0^{K} \neq\varnothing \; (= \varnothing)$
(and vice versa for $K \cap I_{\vert H\vert}^{H}$).  
Taking this together with Lemma~\ref{lemmamoresepfunc},
the coefficient of $T^{}_K$ in \eqref{toTGphi} turns into
\begin{equation}\label{importantz}
\begin{split}
 \sum_{\substack{ H \subseteq L\setminus K \\ \vert H \vert \geq 2} }
             (-1)^{\vert H \vert} \rho^{}_H 
             &\sum_{ G^{'}\subseteq K} (-1)^{\vert G^{'}\vert}\sep(G^{'},H)  \nonumber \\
 & =\, - \sum_{\substack{ H \subseteq L\setminus K \\ \vert H \vert \geq 2} } 
       \rho^{}_H \sep(K,H)\, (1- \delta^{}_{H\cap I_0^{K},\varnothing})
         (1-\delta^{}_{H\cap I_{\vert K\vert}^{K},\varnothing}) \nonumber \\
&  =\, - \sum_{H\subseteq L\setminus K} 
       \rho^{}_H \sep(K,H)\, (1- \delta^{}_{H\cap I_0^{K},\varnothing})
         (1-\delta^{}_{H\cap I_{\vert K\vert}^{K},\varnothing}) \, .
\end{split}
\end{equation}
Note that, in the last step, the restriction on $\vert H\vert$ may be dropped
since it is already implied by the factors involving the $\delta$-functions.
This proves the claim for $G = \varnothing$.
For the case $G \neq \varnothing$, we follow Proposition~\ref{thm:prodLDEops}
and write, for $p \in \mathcal{P}(X)$,
 \[
          T_G^{(L)} \bigl( \varPhi^{ ( L ) } (p) \bigr) \, =\,
         \Bigl(T_{ \varnothing}^{( I^{}_0)} \bigl( \varPhi^{ ( I^{}_0)} ( \pi^{}_{J^{}_0 \cdot} p) \bigr)\Bigr)
         \otimes \Bigl( T_{ \varnothing}^{ (I^{}_1)}\bigl( \varPhi^{( I^{}_1)} ( \pi^{}_{J^{}_1 \cdot} p ) \bigr) \Bigr)
         \otimes \cdots \otimes 
        \Bigl( T_{ \varnothing}^{( I^{}_{\vert G\vert})}
         \bigl( \varPhi^{( I^{}_{\vert G\vert}) } ( \pi^{}_{J^{}_{\vert G\vert} \cdot} p ) \bigr)\Bigr) \, .
      \]
Applying the above result for  $G = \varnothing$ to each factor, 
and using the product structure of Proposition~\ref{thm:prodLDEops} backwards,
establishes the claim. \qed
\end{proof}

\begin{coro}
 The coefficients $z(\varnothing, K)$ with  $K \neq \varnothing$ can be expressed explicitly as 
\begin{equation}\label{zsecond}
   z^{(L)} ( \varnothing , K)\, =\, -\sum_{ \alpha^{}_0 \in I_0^K} \rho^{}_{ \alpha^{}_0} 
   \bigl( \prod_{i=1}^{ \vert K \vert -1} (1 + \sum_{\alpha^{}_i \in I_i^K} \rho^{}_{\alpha^{}_i})\bigr)
   \sum_{\alpha^{}_{ \vert K \vert} \in I_{ \vert K \vert}^K} \rho^{}_{\alpha^{}_{ \vert K \vert}} \, . 
\end{equation}
\end{coro}
\begin{proof}
Let us consider those $H$ whose contribution to the sum in \eqref{zleerk}
is not annihilated by the separation function or
the $\delta$-functions.
For $\sep(K,H) =1$ to hold, each $\alpha \in H$ must belong to a different $I^K_i\in\mathcal{L}^{}_K$.
Furthermore, $H$ must contain one element each from $I^{K}_0$ and $I^{K}_{\vert K\vert}$
($\alpha^{}_0$ and $\alpha^{}_{\vert K\vert}$, respectively) to keep the factors
involving the $\delta$-functions from vanishing. Thus, the sum in \eqref{zleerk} may be factorised as claimed.\qed
\end{proof}

In particular, $z^{(L)} ( \varnothing , K) = 0$
if $K \cap \left\{\frac{1}{2},\frac{2n-1}{2}\right\}\neq\varnothing$.
Taking this together with \eqref{zleerleer}, one obtains
$ z^{( L)} ( \varnothing, K) = (1- \sum_{\alpha \in L}\rho^{}_{\alpha})\thinspace\delta^{}_{K,\varnothing}$
for $K\subseteq L$ whenever $\vert L\vert \leq 2$, and hence, in these cases,
$T^{(L)}_{\varnothing} \circ \varPhi^{(L)} = (1- \sum_{\alpha \in L}\rho^{}_{\alpha}) T^{(L)}_{\varnothing}$
is already a  diagonal component in line with the observation in Section~4.
Furthermore, Theorem~\ref{thm:tvarnotcomplete} and \eqref{zsecond} entail that
$ z^{( L)} ( G, K) = 0$ whenever
\begin{equation}\label{zeqzero}
K \cap \Bigl(\bigcup_{0\leq i\leq\vert G\vert} \{ \min (I_i^G),
\max (I_i^G)\} \Bigr) \neq \varnothing \, .
\end{equation}
Theorem~\ref{thm:tvarnotcomplete} reveals the linear structure inherent in the action of
$T^{}_{G}$ on $\varPhi$.
In fact, the structure is even triangular (with respect to the partial ordering)
since $T^{(L)}_G \circ \varPhi^{(L)}$ is a linear combination of the
$T^{(L)}_K$, $K\supseteq G$.
Thus, diagonalisation will boil down to recursive elimination.
As a preparation, we make the following observation.
\begin{coro}\label{corohelp}
If $L \neq\varnothing$, one has the relation $z^{(L)}(G,L) = 0$ for all
$\varnothing\subseteq G\subsetneq L$.
\end{coro}
\begin{proof}
 When $\varnothing\subseteq G\subsetneq L$, the intersection in \eqref{zeqzero}, with
$K=L$, can never be empty, so that $z^{(L)}(G,L) = 0$ follows. \qed
\end{proof}

\section{Diagonalisation}\label{sec:diagonal}
Motivated by the triangular structure of \eqref{TGonPhi}, we make the ansatz
to define new operators $U^{}_G$, $G\subseteq L$, as the following linear combination
of the well-known $T^{}_G$:
\begin{equation}\label{newtrafo}
    U^{}_G \,=\, \sum_{ H \supseteq G} c (G , H )\, T^{}_H  \, ,
\end{equation}
where the coefficients $c (G , H )$ are to be determined in such a way that 
they transform the recombination equation into a decoupled diagonal system,
more precisely so that
\begin{equation}\label{aimrecur}
U^{}_G \circ \varPhi \, =\, \lambda^{}_G U^{}_G \, , \quad G\subseteq L \, ,
\end{equation}
with eigenvalues $\lambda^{}_G$ that are still unknown as well.
An example for this transformation can be found in Appendix~A.
Note first that, with the help of \eqref{TGonPhi}, 
Eqs.~\eqref{newtrafo} and \eqref{aimrecur} may be rewritten as
\begin{equation}\label{eigenprob}
\begin{split}
   U^{}_G \circ \varPhi \,&=\, c(G,G)\, T^{}_G \circ \varPhi + 
        \sum_{ N \supsetneq G} c (G,N)\, (T^{}_N \circ \varPhi) \\ 
    &=\, c(G,G) \Bigl( z^{(L)}(G,G)\, T^{}_G + \sum_ {K \supsetneq G} z^{ \left( L \right)} (G,K) \, T^{}_K \Bigr) \\
         &\phantom{=}+ \sum_{ N \supsetneq G } c (G,N) \Bigl( z^{(L)}(N,N)\, T^{}_N
         +\sum_{ M \supsetneq N } z^{ \left (L \right) } (N,M) \, T^{}_M \Bigr) \\
    &\stackrel{!}{=} \,\lambda^{}_G \Bigl( c(G,G)\, T^{}_G + \sum_{ N \supsetneq G} c (G,N) \, T^{}_N \Bigr)
     = \lambda^{}_G U^{}_G .
\end{split}
\end{equation}
Obviously, there is some freedom in the choice of the $c(G,G)$;
we set $c(G,G)=1$ for all $G\subseteq L$
(and we will see shortly that this is consistent).
Eq.~\eqref{eigenprob} has the structure of an eigenvalue problem of a triangular matrix
with coefficients $z^{(L)}(G,H)$, where the role of the unit vectors is taken
by the $T^{}_H$, and the $c(G,H)$, $H\supseteq G$, take the roles of the
components of the eigenvector corresponding to $\lambda^{}_G$
(note that, by considering $c(G,H)$ for $H\supseteq G$ only,
we have already exploited the triangular structure).
Recall next that the eigenvalues of a triangular matrix are given by its
diagonal entries, which are
\begin{equation}\label{def:lambda}
  \lambda^{}_G \,=\, z^{(L)}(G,G)\, =\,  \prod_{ i = 0 }^{ \vert G \vert} z^{(I^{G}_i)}(\varnothing,\varnothing)
              \, =\, \prod_{ i = 0 }^{ \vert G \vert} 
            \Bigl( 1 - \sum_{ \alpha^{}_i \in I_i^G} \rho^{}_{ \alpha^{}_i} \Bigr) 
\end{equation}
by Theorem~\ref{thm:tvarnotcomplete}.
In particular, $\lambda^{}_{\varnothing} = \eta = 1-\sum_{\alpha\in L}\rho^{}_{\alpha}\geq 0$.
The $\lambda^{}_G$ describe the probability that there is no further recombination
between the respective sites of the subsystems corresponding to $G$;
they have already been identified by Bennett~\cite{Bennett} and Dawson~\citep{Kevin1, Kevin2}.
\begin{lemma}\label{lemma:zulambdas}
For all $G,H \subseteq L$ with $G\subsetneq H$, one has
$\lambda^{}_G<\lambda^{}_H$.
 \end{lemma}

\begin{proof}
Let $\varnothing \subsetneq G \subsetneq L$.
Then, for $H = G \stackrel{ \cdot }{ \cup } \left\{ \beta \right\}$,
with $\beta\in I_i^G$ for an arbitrary $i \in \left\{0, \ldots, \vert G \vert \right\}$,
we see from \eqref{def:lambda} that $z^{(L)}(H,H) = \lambda^{}_H$
and hence obtain
\[
\begin{split}
\lambda^{}_H \,&=\, 
  \Biggl( \prod_{j=0}^{i-1} \Bigl( 1 - \sum_{\alpha^{}_j \in I_j^G} \rho^{}_{\alpha^{}_j}\Bigr) \Biggr)
  \, \Bigl( 1 - \sum_{ \substack{\alpha^{}_i \in I_{i}^G \\ \alpha^{}_i<\beta}} \rho^{}_{ \alpha^{}_i} \Bigr)\,
  \Bigl( 1 - \sum_{ \substack{\alpha^{}_i \in I_i^G \\ \alpha^{}_i>\beta}} \rho^{}_{\alpha^{}_i} \Bigr)
  \, \Biggl( \prod_{j=i+1}^{\vert G\vert}
  \Bigl( 1 - \sum_{\alpha^{}_j\in I_j^G}\rho^{}_{\alpha^{}_j}\Bigr)\Biggr) \\ 
&= \lambda^{}_G \,
  \frac{ \Bigl( 1 - \sum_{ \substack{\alpha^{}_i \in I_i^G \\
      \alpha^{}_i<\beta}} \rho^{}_{\alpha^{}_i}\Bigr) 
      \Bigl( 1 - \sum_{ \substack{\alpha^{}_i \in I_i^G \\
      \alpha^{}_i>\beta}} \rho^{}_{\alpha^{}_i}\Bigr)}
  {\left( 1 - \sum_{\alpha^{}_i \in I_i^G} \rho^{}_{\alpha^{}_i} \right)} > \lambda^{}_G \, ,
\end{split}
\]
because all $\rho^{}_{\alpha}$ are positive,
as are all three terms in parentheses of the fraction,
and $\rho^{}_{\beta}>0$ by assumption. 
Finally, the argument also works for $\lambda^{}_{\varnothing} = \eta$, provided
$\eta>0$.
Since $\lambda^{}_G >0$ for all $G\neq \varnothing$, the claim trivially also holds for $\eta =0$.
The assertion then follows inductively for any $H \supsetneq G$. \qed
\end{proof}

The coefficients $c(G,H)$ can now be calculated recursively as follows.
\begin{theorem}\label{sumTheorem}
The coefficients $c(G,H)$ of \eqref{newtrafo} are determined by
$c(G,G) =1$ and
\begin{equation}\label{ckoeffrecursion}
    c ( G , H )\, =\, \frac{ 
            \sum_ {H \supsetneq K \supseteq G} c (G , K )\,
                                 z^{ ( L )} (K , H )}{ \lambda^{}_G - \lambda^{}_H} 
\end{equation}
for $H\supsetneq G$.
The coefficients of the inverse transformation of \eqref{newtrafo}, 
\begin{equation}\label{fromTtoU}
 T^{}_G \,=\, \sum_{ H \supseteq G} c^{*} ( G , H )\, U^{}_H \, , \quad\ \text{with}\ G \subseteq L \, ,
\end{equation}
are determined by 
\begin{equation}\label{cstarkoeff}
   c^* ( G , K )\, =\, - \sum_{ K \supsetneq H \supseteq G} c^* ( G , H )\, c ( H , K ) \, ,
\end{equation}
for $K\supsetneq G$ together with $c^{*} ( G , G ) = 1$.
\end{theorem}
\begin{proof}
Considering \eqref{eigenprob} with $c (G, G) = 1$,
comparing coefficients for $T^{}_H$, $H\supsetneq G$, 
and observing \eqref{def:lambda}, one obtains
\begin{equation*}
  z^{ (L ) } (G , H) + c (G, H) \lambda^{}_H
    + \sum_{ H \supsetneq K \supsetneq G } c (G , K)\, z^{ ( L ) } (K, H) 
       \,\stackrel{!}{=}\, \lambda^{}_G \, c (G, H) \, ,
\end{equation*}
and the recursion for $c(G,H)$ follows. It is always well-defined 
for all $H \supsetneq G$, since 
$\lambda^{}_G < \lambda^{}_H$ by Lemma~\ref{lemma:zulambdas}.
The recursion for the coefficients of the inverse transformation
follows directly from
\begin{equation*}
T^{}_G \,=\, \sum_{ H \supseteq G } c^{*} (G,H)\, U^{}_H \,=\, 
      \sum_{ H \supseteq G }c^{*} (G,H) \sum_{ K\supseteq H } c(H,K)\, T^{}_K \,=\,
      \sum_{ K \supseteq G} T^{}_K \sum_{ K \supseteq H \supseteq G} c^{*} (G,H)\, c(H,K) \, ,
\end{equation*}
which enforces $\sum_{ K \supseteq H \supseteq G } c^{*} (G,H)\, c(H,K) = \delta^{}_{K,G}$,
as the $T^{}_K$ are distinct. \qed
\end{proof}

We now identify those $T^{}_G$ 
that already give diagonal components of the discrete-time system:

\begin{theorem}\label{thm:direct}
    For all $G \subseteq L$ that satisfy 
    $\vert I_i^G \vert \leq 2$ for all $i \in \left\{ 0, \ldots , \vert G \vert \right\}$,
    one has
    \[
    T^{}_G \bigl( \varPhi (p) \bigr) \,=\, \lambda^{}_G T^{}_G(p) 
    \]
    for $p\in\mathcal{P}(X)$.
  \end{theorem}

\begin{proof}
   In this case, we have \eqref{zeqzero} for all $K\supsetneq G$,
   hence $z(G,K)=\lambda^{}_G \delta^{}_{K,G}$, from which the assertion follows via
   Theorem~\ref{thm:tvarnotcomplete}. \qed
\end{proof}
Note that  $\vert I_i^G \vert \leq 2$ for all $I_i^G \in \mathcal{L}^{}_G$ simply implies that each subsystem consists of
at most three sites, hence all subsystems can be reduced to the simple cases considered in Section~\ref{SCRdiscrete}.
Then, for such $G$, $c(G,H) = c^*(G,H) = \delta^{}_{G,H}$ for all $H\supseteq G$.

With the help of this transformation, we can finally specify 
the solution $p^{}_t$ of the recombination equation
in terms of the initial condition $p^{}_0$.
To this end, we first use the transformation~\eqref{trafo-2} from 
the recombinators to the $T^{}_G$ operators, and then relation~\eqref{fromTtoU} 
to arrive at the $U^{}_H$ operators, which finally diagonalise the system
according to \eqref{aimrecur}.
Finally, we use the appropriate inversions to return to
the recombinators:
\begin{equation}\label{finalsol}
\begin{split}
p^{}_t \,&=\, \varPhi^t (p^{}_0) \,=\, R^{}_{ \varnothing } ( \varPhi^t (p^{}_0))\, =\,
                \sum_{ G \subseteq L} T^{}_G ( \varPhi^t (p^{}_0 )) \,=\,
               \sum_{ G \subseteq L } \sum_{ H \supseteq G} c^* (G,H )\, U^{}_H ( \varPhi^t (p^{}_0)) \\ &=
               \sum_{ G \subseteq L} \sum_{ H \supseteq G } c^* (G,H) \lambda_H^t \,U^{}_H (p^{}_0) \,=\,
               \sum_{ G \subseteq L } \sum_{ H \supseteq G } c^* (G,H)
               \lambda_H^t \sum_{ M \supseteq H } c (H,M)\, T^{}_M (p^{}_0) \\ &=
               \sum_{ G \subseteq L } \sum_{ H \supseteq G} c^* (G,H) \lambda_H^t
             \sum_{ M \supseteq H } c (H,M) \sum_{ T \supseteq M } (-1)^{ \left| T-M \right|} R^{}_T (p^{}_0) \, .
\end{split}
\end{equation}
The coefficient functions can now be extracted as follows.
\begin{theorem} \label{thm:aexplicit}
The coefficient functions $a^{}_G ( t )$ of the solution~\eqref{solutiondiscrete}
of the recombination equation in discrete time may be expressed as 
\[
    a^{}_G (t) \,= \sum_{ M \subseteq G} (-1)^{ \vert G - M \vert} \sum_{ H \subseteq M}
         \sum_{ K \subseteq H}  c ( H , M )\,\lambda_H^t \, c^* ( K , H )  
\]
for all $G \subseteq L$.
Here, $c(H,M)$ and $c^{*}(K,H)$ are the coefficients of Theorem~\ref{sumTheorem}. \qed
\end{theorem}

To derive the asymptotic behaviour for large iteration numbers, we need the following
property of the coefficients.
\begin{lemma}\label{cforL}
 The coefficients $c(G,L)$ and $c^{*}(G,L)$ satisfy 
$c(G,L) = c^{*}(G,L) = \delta^{}_{G,L}$
for arbitrary $\varnothing\subseteq G\subseteq L$.
\end{lemma}
\begin{proof}
We have $c(G,G)=c^{*}(G,G) =1$ for all $G$ by Theorem~\ref{sumTheorem}.
The claim for $c(G,L)$ now follows from the recursion~\eqref{ckoeffrecursion}
together with Corollary~\ref{corohelp}. Inserting this into recursion~\eqref{cstarkoeff}
establishes the relation for the $c^{*}(G,L)$. \qed
\end{proof}

As an example, the path to a solution via the above chain of transformations 
for the model with five sites will be presented in Appendix~A.

Finally, let us consider what happens in the limit as $t\to\infty$.
\begin{prop}
The solution $p^{}_t$ of the recombination equation \emph{\eqref{rekooper}} 
with initial condition $p^{}_0$ satisfies
\[
  p^{}_t \xrightarrow{ t \to \infty} R^{}_L (p^{}_0 ) =\,
    \bigotimes_{ i = 0 }^{n}( \pi^{}_{i \cdot} p^{}_0 )\, ,
\]
with exponentially fast convergence in the norm topology.
\end{prop}
\begin{proof}
When expressing $p^{}_t$ in terms of $U^{}_H$ according to \eqref{finalsol},
we first observe 
$p^{}_t = U^{}_L(p^{}_0) + \sum_{ G \subsetneq L} 
    \sum_{ H \supseteq G}  c^* \left( G, H \right) \lambda_H^t\, U^{}_H (p^{}_0)$,
because $\lambda^{}_L = 1$ and $c^*(G,L)= \delta^{}_{G,L}$  by Lemma~\ref{cforL}.
Since $U^{}_L = R^{}_L$, we obtain the following estimate in the variation norm
\[
\begin{split}
\Vert p^{}_t - R^{}_L (p^{}_0) \Vert &=
\Big\| \sum_{ G \subsetneq L} 
    \sum_{ H \supseteq G} c^* \left( G, H \right) \lambda_H^t U^{}_H (p^{}_0) \Big\| \\
  &\leq  \sum_{H \subsetneq L} \lambda_H^t  \Big\| \sum_{G \subseteq H}
     c^*\left(G,H\right)  U^{}_H (p^{}_0) \Big\| \xrightarrow{t\to\infty}0 \, ,
\end{split}
\]
which establishes the claim since $\lambda^{}_H < 1$ for $H \neq L$. \qed
\end{proof}
As was to be expected, the solution of the recombination equation 
converges towards the independent combination of the alleles, that is
towards \emph{linkage equilibrium}.

\section{Discussion}\label{sec:discuss}
In this paper, we have investigated the dynamics of an `infinite' population 
that evolves due to recombination alone.
To this end, we assumed discrete (non-overlapping) generations, 
and restricted ourselves to the case of single crossovers.
Previous results had shown that the corresponding single-crossover dynamics
in continuous time admits a closed solution \cite{reco}.
This astonishing result is concordant with a `hidden' linearity 
in the system that is due to independence of links.
The fact that crossovers at different links occur independently 
manifests itself in the product structure of the coefficient functions of the solution
ensuing from the linear action of the 
nonlinear recombination operators along the solution of the recombination equation.
Additionally, in \cite{reco}, a certain set of linkage disequilibria
was found that linearise and diagonalise the dynamics.

Since the overwhelming part of the literature deals with discrete-time
models, our aim was to find out whether, and to what extent, these continuous-time
results carry over to single-crossover dynamics in discrete time.
We could show that the discrete-time dynamics is far more complex than
the continuous-time one, and, as a consequence, a closed solution 
cannot be given.

The main reason for these difficulties lies in the fact that the key
feature of the continuous-time model, the independence of links, does not
carry over to discrete time. This is due to interference: 
The occurrence of a recombination event in the discrete-time model  
forbids any further crossovers in the same generation.
In connection with this, the recombinators do not, in general,
act linearly on the right-hand side
of the recombination equation.
Likewise, the coefficient functions of the solution follow a nonlinear
iteration that cannot be solved explicitly.

While Geiringer~\cite{Geiringer} developed a skilful procedure
for the generation-wise evaluation of these coefficients, we constructed
a method that allows for an {\em explicit} formula valid for all times,
once the coefficients of the transformation have been determined
recursively for a given system.

As in previous approaches, this is achieved by a transformation
of the nonlinear, coupled system of equations to a linear diagonal one.
This was done before by Bennett~\cite{Bennett} and
Dawson~\citep{Kevin1, Kevin2} for the more general
recombination equation (without restriction to single crossovers),
and they presented an appropriate transformation 
that includes parameters that must be determined recursively.
Unfortunately, the corresponding derivations are rather technical and
fail to reveal the underlying mathematical structure.
It was our aim to improve on this and add some structural insight.
Unlike the previous approaches, we proceeded in two steps: 
first linearisation followed by diagonalisation.
More precisely, it turns out that the LDE operators $T^{}_G$, which
both linearise and diagonalise the continuous-time system,
still {\em linearise} the discrete-time dynamics, but fail 
to {\em diagonalise} it for four or more loci.
However, the resulting linear system may
be diagonalised in a second step. This relies on linear combinations
$U^{}_G$ of the $T^{}_G$, with coefficients derived in a recursive 
manner.

As it must be, the transformation agrees with the one of Dawson~\citep{Kevin1, Kevin2}
when translated into his framework.
(Note that our $c(G,H)$ are coefficients of $T^{}_H$, whereas his coefficients
belong to components of $R^{}_H(p)$.
Note also that SCR does {\em not} belong to the singular cases he excludes).
It remains an interesting open problem how much of the above findings 
can be transferred to the general recombination model
(i.e. without the restriction to single crossovers),
where one loses the simplifying structure of ordered partitions.

\begin{acknowledgements}
It is our pleasure to thank Th. Hustedt for critically reading
the manuscript and K. Schneider for valuable suggestions
to further improve it.
This work was supported by DFG (Research training group Bioinformatics,
and Dutch-German Bilateral Research Group on 
Mathematics of Random Spatial Models in Physics and Biology).
\end{acknowledgements}

\section*{Appendix A: Five Sites}

To illustrate the construction, let us spell out the example of five sites. 
We have $S=\left\{ 0,1,2,3,4 \right\}$ 
and $ L= \left\{ \frac{1}{2}, \frac{3}{2}, \frac{5}{2}, \frac{7}{2} \right\}$,
the corresponding recombination probabilities $\rho^{}_{\alpha}$, $\alpha\in L$,
$\eta = 1 - \rho^{}_{ \frac{1}{2} } - \rho^{}_{ \frac{3}{2} }
          - \rho^{}_{ \frac{5}{2} } - \rho^{}_{ \frac{7}{2} }$, 
and a given initial population $p^{}_0$.
Aiming at determining the coefficient functions $a^{}_G(t)$ for all $G\subseteq L$,
we can immediately write down $a^{}_{ \varnothing } (t) = \eta^t$,
$a^{}_{ \frac{1}{2} } (t) = (\eta + \rho^{}_{\frac{1}{2}})^t- \eta^t$,
$a^{}_{ \frac{7}{2} } (t) = (\eta + \rho^{}_{\frac{7}{2}})^t - \eta^t$ and 
$a^{}_{ \left\{ \frac{1}{2}, \frac{7}{2}\right\} } (t) =
             \eta^t - (\eta + \rho^{}_{\frac{1}{2}})^t - (\eta + \rho^{}_{\frac{7}{2}})^t
            + (\eta + \rho^{}_{\frac{1}{2}} + \rho^{}_{\frac{7}{2}})^t$,
see \eqref{adirectly}.

If we wanted to determine the remaining coefficient functions  $a^{}_G(t)$ 
for a given time $t$,
they could be calculated using the method of Geiringer~\cite{Geiringer}
(i.e. Theorem~\ref{thm:adevelop}).
But since we aim at a closed solution for {\em all} $t$, 
we use the procedure developed above.
To determine the coefficients of Theorem~\ref{thm:aexplicit},
we have to calculate the corresponding  $c(G,H)$ and
$c^*(G,H)$.
Theorem~\ref{sumTheorem}~and~\ref{thm:direct} imply $U^{}_L = T^{}_L$, 
$U^{}_{L \setminus \{ \alpha \}} = T^{}_{L \setminus \{ \alpha \}}$ 
for all $\alpha \in L$, 
$U^{}_{L \setminus \{ \alpha, \beta \}} = T^{}_{L \setminus \{ \alpha, \beta \}}$
for all $\alpha, \beta \in L$, as well as
$U^{}_{ \frac{3}{2}} = T^{}_{ \frac{3}{2}}$ and 
$U^{}_{ \frac{5}{2}} = T^{}_{ \frac{5}{2}}$.
Hence, in these cases, the only non-vanishing coefficients are
$c(L,L) = c(L \setminus \{ \alpha \},L \setminus \{ \alpha \})
= c(L \setminus \{ \alpha, \beta \}, L \setminus \{ \alpha, \beta \})
= c(\{\tfrac{3}{2}\}, \{ \tfrac{3}{2} \}) = 
c(\{ \tfrac{5}{2} \}, \{ \tfrac{5}{2}\} ) = 1$ for all $\alpha,\beta\in L$.
It remains to determine $U^{}_{ \frac{1}{2}}$, $U^{}_{ \frac{7}{2}}$
and $U^{}_{ \varnothing}$.

\begin{enumerate}
 \item {\em Constructing $U^{}_{ \frac{1}{2}}$}:

The recursion starts with $c(\{\tfrac{1}{2}\},\{\tfrac{1}{2}\}) = 1$. 
Following \eqref{zeqzero}, $z^{(L)}(\{ \tfrac{1}{2} \}, H) =0$
for all $H\supsetneq \{\tfrac{1}{2}\}$ except for 
$H = \{\tfrac{1}{2},\tfrac{5}{2}\}$,
and thus the only non-zero $c(\{\tfrac{1}{2}\},H)$, $H\supsetneq \{\tfrac{1}{2}\}$,
is 
\[
    c(\{\tfrac{1}{2}\},\{\tfrac{1}{2},\tfrac{5}{2}\}) =
      \frac{ 
           z( \{\frac{1}{2}\}, \{\frac{1}{2},\frac{5}{2}\} ) }
         { \lambda^{}_{ \frac{1}{2}  } - \lambda^{}_{ \{ \frac{1}{2}, \frac{5}{2}\} } } =
         \frac{ 
             \rho^{}_{ \frac{3}{2} }  \rho^{}_{ \frac{7}{2} } }
         { \rho^{}_{ \frac{5}{2} } + \rho^{}_{ \frac{3}{2}}  \rho^{}_{ \frac{7}{2}} } \, ,
\]
where we have used the recursion~\eqref{ckoeffrecursion} and 
$\lambda^{}_{ \frac{1}{2} } = 1-\rho^{}_{\frac{3}{2}}-\rho^{}_{\frac{5}{2}}
- \rho^{}_{\frac{7}{2}}$,
$\lambda^{}_{ \{\frac{1}{2},\frac{5}{2} \}}~ =~ 
(1~-~\rho^{}_{\frac{3}{2}})~(1~-~\rho^{}_{\frac{7}{2}})$.
So, for the transformation~\eqref{newtrafo} we obtain 
\[
   U^{}_{ \frac{1}{2} } =
      T^{}_{ \frac{1}{2} } + \frac{ \rho^{}_{ \frac{3}{2} } 
             \rho^{}_{ \frac{7}{2} } }
         { \rho^{}_{ \frac{5}{2} } + \rho^{}_{ \frac{3}{2} }
             \rho^{}_{\frac{7}{2}} }
        T^{}_{\{ \frac{1}{2},\frac{5}{2}\} } \, ,
\]
so that $U^{}_{ \frac{1}{2}}\circ\varPhi = ( 1-\rho^{}_{\frac{3}{2}}-\rho^{}_{\frac{5}{2}}
- \rho^{}_{\frac{7}{2}}) U^{}_{\frac{1}{2}}$.
Analogously, 
\[
    U^{}_{ \frac{7}{2} } =
      T^{}_{ \frac{7}{2} } + \frac{ \rho^{}_{ \frac{1}{2} } 
                  \rho^{}_{\frac{5}{2}}}
           { \rho^{}_{ \frac{3}{2} } + \rho^{}_{ \frac{1}{2} } 
        \rho^{}_{ \frac{5}{2} } }
                 T^{}_{ \{ \frac{3}{2}, \frac{7}{2} \} } \, .
\]

\item {\em Constructing $U^{}_{\varnothing}$}:

By \eqref{zeqzero}, the only non-vanishing coefficients are 
$c(\varnothing, \varnothing)$, $c(\varnothing, \{ \frac{3}{2} \})$,
$c(\varnothing, \{ \frac{5}{2} \})$, and 
$c(\varnothing, \{ \frac{3}{2}, \frac{5}{2} \})$. 
They are determined by
the recursion~\eqref{ckoeffrecursion} and lead to the following
transformation~\eqref{newtrafo}:
\[
  U^{}_{ \varnothing } =
            T^{}_{\varnothing} +
    \frac{
   \rho^{}_{ \frac{1}{2} } ( \rho^{}_{ \frac{5}{2} } + \rho^{}_{ \frac{7}{2} } )  }
     { \rho^{}_{ \frac{3}{2} } + \rho^{}_{ \frac{1}{2} } 
          ( \rho^{}_{ \frac{5}{2} } + \rho^{}_{ \frac{7}{2} } ) }
              T^{}_{ \frac{3}{2} } +
       \frac{
          ( \rho^{}_{ \frac{1}{2} } + \rho^{}_{ \frac{3}{2} } ) \rho^{}_{ \frac{7}{2} } }
             { \rho^{}_{ \frac{5}{2} } +
          ( \rho^{}_{ \frac{1}{2} } + \rho^{}_{ \frac{3}{2} } ) \rho^{}_{ \frac{7}{2} } }
              T^{}_{ \frac{5}{2} } +
       \frac{ \rho^{}_{ \frac{1}{2} }  \rho^{}_{ \frac{7}{2} } }
       { \rho^{}_{ \frac{1}{2} }  \rho^{}_{ \frac{7}{2} } +
      \rho^{}_{ \frac{3}{2} } + \rho^{}_{ \frac{5}{2} } }
              T^{}_{ \left\{ \frac{3}{2}, \frac{5}{2} \right\} } \, .
\]
\end{enumerate}

Now that we know the $c(G, H)$, the coefficients $c^*(G,H)$ 
are calculated via \eqref{cstarkoeff}.
Finally, the remaining coefficient functions follow from  Theorem~\ref{thm:aexplicit}:
\[
 \begin{split}
    a^{}_{\frac{3}{2}} (t) &= \frac{ \rho^{}_{ \frac{3}{2} } }
          {\rho^{}_{\frac{1}{2}}(\rho^{}_{\frac{5}{2}}+\rho^{}_{\frac{7}{2}})
                 + \rho^{}_{\frac{3}{2}}}
       (\lambda_{\frac{3}{2}}^t - \lambda_{\varnothing}^t) \\
   a^{}_{\frac{5}{2}} (t) &= \frac{ \rho^{}_{ \frac{5}{2} } }
          {\rho^{}_{\frac{7}{2}}(\rho^{}_{\frac{1}{2}}+\rho^{}_{\frac{3}{2}})
                 + \rho^{}_{\frac{5}{2}}}
       (\lambda_{\frac{5}{2}}^t - \lambda_{\varnothing}^t) \\
   a^{}_{\left\{\frac{1}{2},\frac{3}{2}\right\}} (t) &=
         \lambda_{ \{ \frac{1}{2}, \frac{3}{2} \}}^t - \lambda_{ \frac{1}{2}}^t
              - \frac{ \rho^{}_{ \frac{3}{2} } }
          {\rho^{}_{\frac{1}{2}}(\rho^{}_{\frac{5}{2}}+\rho^{}_{\frac{7}{2}})
                 + \rho^{}_{\frac{3}{2}}}
           ( \lambda_{ \frac{3}{2} }^t - \lambda_{ \varnothing}^t ) \\
  a^{}_{\left\{\frac{1}{2},\frac{5}{2}\right\}} (t) &=
         \frac{ \rho^{}_{ \frac{5}{2} } }
        { \rho^{}_{ \frac{3}{2} } \rho^{}_{ \frac{7}{2} } + \rho^{}_{ \frac{5}{2} } }
     ( \lambda_{ \{ \frac{1}{2},\frac{5}{2}\} }^t - \lambda_{ \frac{1}{2}}^t )        
   - \frac{ \rho^{}_{ \frac{5}{2}}} 
  {\rho^{}_{ \frac{7}{2} } ( \rho^{}_{ \frac{1}{2}} + \rho^{}_{\frac{3}{2}} )
      + \rho^{}_{ \frac{5}{2}} } ( \lambda_{ \frac{5}{2} }^t - \lambda_{\varnothing }^t )  \\
 a^{}_{\left\{\frac{3}{2},\frac{5}{2}\right\}} (t) &=
       \frac{ \rho^{}_{ \frac{3}{2} } + \rho^{}_{ \frac{5}{2} } }
        { \rho^{}_{ \frac{1}{2} } \rho^{}_{ \frac{7}{2} } 
           + \rho^{}_{ \frac{3}{2} } +  \rho^{}_{ \frac{5}{2} } }
        \lambda_{ \{ \frac{3}{2}, \frac{5}{2} \}}^t
     - \frac{ \rho^{}_{ \frac{3}{2}}} 
  {\rho^{}_{ \frac{1}{2} } ( \rho^{}_{ \frac{5}{2}} + \rho^{}_{\frac{7}{2}} )
      + \rho^{}_{ \frac{3}{2}} } \lambda_{\frac{3}{2}}^t
     - \frac{ \rho^{}_{ \frac{5}{2}}} 
  {\rho^{}_{ \frac{7}{2} } ( \rho^{}_{ \frac{1}{2}} + \rho^{}_{\frac{3}{2}} )
      + \rho^{}_{ \frac{5}{2}} } \lambda_{\frac{5}{2}}^t \\
   &\phantom{=} + \Bigl(1 - 
           \frac{ ( \rho^{}_{ \frac{1}{2}} + \rho^{}_{ \frac{3}{2}} ) \rho^{}_{ \frac{7}{2}} } 
     {\rho^{}_{ \frac{7}{2} } ( \rho^{}_{ \frac{1}{2}} + \rho^{}_{\frac{3}{2}} )
        + \rho^{}_{ \frac{5}{2}} }
     - \frac{ ( \rho^{}_{ \frac{5}{2}} + \rho^{}_{ \frac{7}{2}} ) \rho^{}_{ \frac{1}{2}} } 
     {\rho^{}_{ \frac{1}{2} } ( \rho^{}_{ \frac{5}{2}} + \rho^{}_{\frac{7}{2}} )
        + \rho^{}_{ \frac{3}{2}} }
       + \frac{ \rho^{}_{ \frac{1}{2}} \rho^{}_{ \frac{7}{2}} } 
        { \rho^{}_{ \frac{1}{2} } \rho^{}_{ \frac{7}{2}} + \rho^{}_{\frac{3}{2}} 
         + \rho^{}_{ \frac{5}{2}} } \Bigr) \lambda_{\varnothing}^t \\
a^{}_{\left\{\frac{3}{2},\frac{7}{2}\right\}} (t) &=
        \frac{ \rho^{}_{ \frac{3}{2} } }
        { \rho^{}_{ \frac{1}{2} } \rho^{}_{ \frac{5}{2} } + \rho^{}_{ \frac{3}{2} } }
     ( \lambda_{ \{ \frac{3}{2},\frac{7}{2}\} }^t - \lambda_{ \frac{7}{2}}^t )
     - \frac{ \rho^{}_{ \frac{3}{2}}} 
  {\rho^{}_{ \frac{1}{2} } ( \rho^{}_{ \frac{5}{2}} + \rho^{}_{\frac{7}{2}} )
      + \rho^{}_{ \frac{3}{2}} } (\lambda_{ \frac{3}{2} }^t - \lambda_{ \varnothing }^t ) \\
 a^{}_{\left\{\frac{5}{2},\frac{7}{2}\right\}} (t) &=
         \lambda_{ \{ \frac{5}{2}, \frac{7}{2} \}}^t - \lambda_{ \frac{7}{2}}^t
              - \frac{ \rho^{}_{ \frac{5}{2} } }
          {\rho^{}_{\frac{7}{2}}(\rho^{}_{\frac{1}{2}}+\rho^{}_{\frac{3}{2}})
                 + \rho^{}_{\frac{5}{2}}}
           ( \lambda_{ \frac{5}{2} }^t - \lambda_{ \varnothing}^t ) \\
\end{split}
\]
\[
\begin{split}
a^{}_{\left\{ \frac{1}{2}, \frac{3}{2},\frac{5}{2}\right\}} (t) &=
      \lambda_{\left\{ \frac{1}{2}, \frac{3}{2},\frac{5}{2}\right\}}^t
       - \lambda_{ \left\{ \frac{1}{2}, \frac{3}{2}\right\}}^t
       -  \frac{ \rho^{}_{ \frac{5}{2} } }
          {\rho^{}_{\frac{5}{2}} + \rho^{}_{\frac{3}{2}} \rho^{}_{\frac{7}{2}} } 
       ( \lambda_{ \left\{ \frac{1}{2},\frac{5}{2} \right\}}^t - \lambda_{\frac{1}{2}}^t ) 
        - \frac{ \rho^{}_{ \frac{3}{2} } + \rho^{}_{ \frac{5}{2} }}
          {\rho^{}_{\frac{3}{2}} + \rho^{}_{\frac{5}{2}} +
                       \rho^{}_{\frac{1}{2}} \rho^{}_{\frac{7}{2}} } 
                \lambda_{\left\{ \frac{3}{2}, \frac{5}{2} \right\}}^t  \\
      &\phantom{=} + \frac{ \rho^{}_{ \frac{3}{2}} } 
     {\rho^{}_{ \frac{1}{2} } ( \rho^{}_{ \frac{5}{2}} + \rho^{}_{\frac{7}{2}} )
        + \rho^{}_{ \frac{3}{2}} } \lambda_{\frac{3}{2}}^t  
         +  \frac{  \rho^{}_{ \frac{5}{2}} } 
     {\rho^{}_{ \frac{7}{2} } ( \rho^{}_{ \frac{1}{2}} + \rho^{}_{\frac{3}{2}} )
        + \rho^{}_{ \frac{5}{2}} }  
           \lambda_{\frac{5}{2}}^t  \\
     &\phantom{=} -  \Bigl( 1 - 
           \frac{ ( \rho^{}_{ \frac{1}{2}} + \rho^{}_{ \frac{3}{2}} ) \rho^{}_{ \frac{7}{2}} } 
     {\rho^{}_{ \frac{7}{2} } ( \rho^{}_{ \frac{1}{2}} + \rho^{}_{\frac{3}{2}} )
        + \rho^{}_{ \frac{5}{2}} }  
        - \frac{ ( \rho^{}_{ \frac{5}{2}} + \rho^{}_{ \frac{7}{2}} ) \rho^{}_{ \frac{1}{2}} } 
     {\rho^{}_{ \frac{1}{2} } ( \rho^{}_{ \frac{5}{2}} + \rho^{}_{\frac{7}{2}} )
        + \rho^{}_{ \frac{3}{2}} }  
        + \frac{ \rho^{}_{ \frac{1}{2}} \rho^{}_{ \frac{7}{2}} } 
     {\rho^{}_{ \frac{1}{2} } \rho^{}_{ \frac{7}{2}} + \rho^{}_{\frac{3}{2}} 
        + \rho^{}_{ \frac{5}{2}} }  \Bigr) \lambda_{\varnothing}^t\\
 a^{}_{\left\{ \frac{1}{2}, \frac{3}{2},\frac{7}{2}\right\}} (t) &=
       \lambda_{\left\{ \frac{1}{2}, \frac{3}{2},\frac{7}{2}\right\}}^t
       - \lambda_{ \left\{ \frac{1}{2}, \frac{3}{2}\right\}}^t
       - \lambda_{ \left\{ \frac{1}{2}, \frac{7}{2}\right\}}^t
       - \frac{ \rho^{}_{ \frac{3}{2} } }
          {\rho^{}_{\frac{3}{2}} + \rho^{}_{\frac{1}{2}} \rho^{}_{\frac{5}{2}} } 
           \lambda_{ \left\{ \frac{3}{2}, \frac{7}{2}\right\}}^t
       + \lambda_{\frac{1}{2}}^t \\
       &\phantom{=} +   \frac{ \rho^{}_{ \frac{3}{2} } }
          {\rho^{}_{\frac{3}{2}} +  \rho^{}_{\frac{1}{2}}
                (\rho^{}_{\frac{5}{2}} + \rho^{}_{\frac{7}{2}}) } \lambda_{\frac{3}{2}}^t
       +   \frac{ \rho^{}_{ \frac{3}{2} } }
          {\rho^{}_{\frac{3}{2}} +  \rho^{}_{\frac{1}{2}}
                \rho^{}_{\frac{5}{2}}} \lambda_{\frac{7}{2}}^t
       -  \frac{ \rho^{}_{ \frac{3}{2} } }
          {\rho^{}_{\frac{3}{2}} +  \rho^{}_{\frac{1}{2}}
                (\rho^{}_{\frac{5}{2}} + \rho^{}_{\frac{7}{2}}) } \lambda_{\varnothing}^t \\
 a^{}_{\left\{ \frac{1}{2}, \frac{5}{2},\frac{7}{2}\right\}} (t) &=
       \lambda_{\left\{ \frac{1}{2}, \frac{5}{2},\frac{7}{2}\right\}}^t
       - \lambda_{ \left\{ \frac{5}{2}, \frac{7}{2}\right\}}^t
       - \lambda_{ \left\{ \frac{1}{2}, \frac{7}{2}\right\}}^t
       - \frac{ \rho^{}_{ \frac{5}{2} } }
          {\rho^{}_{\frac{5}{2}} + \rho^{}_{\frac{3}{2}} \rho^{}_{\frac{7}{2}} } 
        \lambda_{ \left\{ \frac{1}{2}, \frac{5}{2}\right\}}^t
       + \lambda_{\frac{7}{2}}^t \\
       &\phantom{=} +   \frac{ \rho^{}_{ \frac{5}{2} } }
          {\rho^{}_{\frac{5}{2}} +  \rho^{}_{\frac{3}{2}} \rho^{}_{\frac{7}{2}} }
      \lambda_{\frac{1}{2}}^t
      +  \frac{ \rho^{}_{ \frac{5}{2} } }
          {\rho^{}_{\frac{5}{2}} +  \rho^{}_{\frac{7}{2}}
                (\rho^{}_{\frac{1}{2}} + \rho^{}_{\frac{3}{2}}) } \lambda_{\frac{5}{2}}^t
       -  \frac{ \rho^{}_{ \frac{5}{2} } }
          {\rho^{}_{\frac{5}{2}} +  \rho^{}_{\frac{7}{2}}
                (\rho^{}_{\frac{1}{2}} + \rho^{}_{\frac{3}{2}}) } \lambda_{\varnothing}^t \\
 \end{split}
\]
\[
\begin{split}
     a^{}_{\left\{ \frac{3}{2}, \frac{5}{2},\frac{7}{2}\right\}} (t) &=
      \lambda_{\left\{ \frac{3}{2}, \frac{5}{2},\frac{7}{2}\right\}}^t
       - \lambda_{ \left\{ \frac{5}{2}, \frac{7}{2}\right\}}^t
       -  \frac{ \rho^{}_{ \frac{3}{2} } }
          {\rho^{}_{\frac{3}{2}} + \rho^{}_{\frac{1}{2}} \rho^{}_{\frac{5}{2}} } 
       ( \lambda_{ \left\{ \frac{3}{2},\frac{7}{2} \right\}}^t - \lambda_{\frac{7}{2}}^t  ) 
        - \frac{ \rho^{}_{ \frac{3}{2} } + \rho^{}_{ \frac{5}{2} }}
          {\rho^{}_{\frac{3}{2}} + \rho^{}_{\frac{5}{2}} +
                       \rho^{}_{\frac{1}{2}} \rho^{}_{\frac{7}{2}} }  
           \lambda_{\left\{ \frac{3}{2}, \frac{5}{2} \right\}}^t\\
      &\phantom{=} +  
           \frac{ \rho^{}_{ \frac{3}{2}} } 
     {\rho^{}_{ \frac{1}{2} } ( \rho^{}_{ \frac{5}{2}} + \rho^{}_{\frac{7}{2}} )
        + \rho^{}_{ \frac{3}{2}} }  \lambda_{\frac{3}{2}}^t 
      +  \frac{  \rho^{}_{ \frac{5}{2}} } 
     {\rho^{}_{ \frac{7}{2} } ( \rho^{}_{ \frac{1}{2}} + \rho^{}_{\frac{3}{2}} )
        + \rho^{}_{ \frac{5}{2}} }   \lambda_{\frac{5}{2}}^t \\
     &\phantom{=} -  \Bigl( 1 - 
           \frac{ ( \rho^{}_{ \frac{1}{2}} + \rho^{}_{ \frac{3}{2}} ) \rho^{}_{ \frac{7}{2}} } 
     {\rho^{}_{ \frac{7}{2} } ( \rho^{}_{ \frac{1}{2}} + \rho^{}_{\frac{3}{2}} )
        + \rho^{}_{ \frac{5}{2}} }  
        - \frac{ ( \rho^{}_{ \frac{5}{2}} + \rho^{}_{ \frac{7}{2}} ) \rho^{}_{ \frac{1}{2}} } 
     {\rho^{}_{ \frac{1}{2} } ( \rho^{}_{ \frac{5}{2}} + \rho^{}_{\frac{7}{2}} )
        + \rho^{}_{ \frac{3}{2}} }  
        + \frac{ \rho^{}_{ \frac{1}{2}} \rho^{}_{ \frac{7}{2}} } 
     {\rho^{}_{ \frac{1}{2} } \rho^{}_{ \frac{7}{2}} + \rho^{}_{\frac{3}{2}} 
        + \rho^{}_{ \frac{5}{2}} }  \Bigr)   \lambda_{\varnothing}^t
   \end{split}
  \]
 and 
\[
\begin{split}
a^{}_{\left\{ \frac{1}{2}, \frac{3}{2}, \frac{5}{2},\frac{7}{2}\right\}} (t) &=
        \lambda_{\left\{ \frac{1}{2}, \frac{3}{2}, \frac{5}{2},\frac{7}{2}\right\}}^t
          - \lambda_{\left\{ \frac{1}{2}, \frac{3}{2}, \frac{5}{2} \right\}}^t
          - \lambda_{\left\{ \frac{1}{2}, \frac{3}{2},\frac{7}{2}\right\}}^t
          - \lambda_{\left\{ \frac{1}{2}, \frac{5}{2},\frac{7}{2}\right\}}^t
          - \lambda_{\left\{ \frac{3}{2}, \frac{5}{2},\frac{7}{2}\right\}}^t
          + \lambda_{\left\{ \frac{1}{2}, \frac{3}{2}\right\}}^t \\
          &\phantom{=} +
            \frac{ \rho^{}_{ \frac{5}{2} } }
          {\rho^{}_{\frac{5}{2}} + \rho^{}_{\frac{3}{2}} \rho^{}_{\frac{7}{2}} } 
          \lambda_{\left\{ \frac{1}{2}, \frac{5}{2}\right\}}^t
         + \lambda_{\left\{ \frac{1}{2}, \frac{7}{2}\right\}}^t
         + \frac{ \rho^{}_{\frac{3}{2}} + \rho^{}_{ \frac{5}{2} } }
          {\rho^{}_{\frac{3}{2}} + \rho^{}_{\frac{5}{2}} + \rho^{}_{\frac{1}{2}} \rho^{}_{\frac{7}{2}} } 
           \lambda_{\left\{ \frac{3}{2}, \frac{5}{2}\right\}}^t
        + \frac{ \rho^{}_{ \frac{3}{2} } }
          {\rho^{}_{\frac{3}{2}} + \rho^{}_{\frac{1}{2}} \rho^{}_{\frac{5}{2}} } 
           \lambda_{\left\{ \frac{3}{2}, \frac{7}{2}\right\}}^t \\
          &\phantom{=} + \lambda_{\left\{ \frac{5}{2}, \frac{7}{2}\right\}}^t 
       -  \frac{ \rho^{}_{ \frac{5}{2} } } 
          {\rho^{}_{\frac{5}{2}} + \rho^{}_{\frac{3}{2}} \rho^{}_{\frac{7}{2}} } \lambda_{\frac{1}{2}}^t
        -  \frac{ \rho^{}_{ \frac{3}{2}} } 
     {\rho^{}_{ \frac{1}{2} } ( \rho^{}_{ \frac{5}{2}} + \rho^{}_{\frac{7}{2}} )
        + \rho^{}_{ \frac{3}{2}} } \lambda_{\frac{3}{2}}^t
        -  \frac{ \rho^{}_{ \frac{5}{2}} } 
     {\rho^{}_{ \frac{7}{2} } ( \rho^{}_{ \frac{1}{2}} + \rho^{}_{\frac{3}{2}} )
        + \rho^{}_{ \frac{5}{2}} }  \lambda_{\frac{5}{2}}^t \\
       &\phantom{=} -  \frac{ \rho^{}_{ \frac{3}{2} } }
          {\rho^{}_{\frac{3}{2}} + \rho^{}_{\frac{1}{2}} \rho^{}_{\frac{5}{2}} } \lambda_{\frac{7}{2}}^t \\
        &\phantom{=} 
+  \Bigl( 1 - 
           \frac{ ( \rho^{}_{ \frac{5}{2}} + \rho^{}_{ \frac{7}{2}} ) \rho^{}_{ \frac{1}{2}} } 
     {\rho^{}_{ \frac{1}{2} } ( \rho^{}_{ \frac{5}{2}} + \rho^{}_{\frac{7}{2}} )
        + \rho^{}_{ \frac{3}{2}} }  
        - \frac{ ( \rho^{}_{ \frac{1}{2}} + \rho^{}_{ \frac{3}{2}} ) \rho^{}_{ \frac{7}{2}} } 
     {\rho^{}_{ \frac{7}{2} } ( \rho^{}_{ \frac{1}{2}} + \rho^{}_{\frac{3}{2}} )
        + \rho^{}_{ \frac{5}{2}} }  
        + \frac{ \rho^{}_{ \frac{1}{2}} \rho^{}_{ \frac{7}{2}} } 
     {\rho^{}_{ \frac{1}{2} } \rho^{}_{ \frac{7}{2}} + \rho^{}_{\frac{3}{2}} 
        + \rho^{}_{ \frac{5}{2}} }  \Bigr) \lambda_{\varnothing}^t \, ,
     \end{split}
\]
where the $\lambda^{}_G$ are given by \eqref{def:lambda}.


\begin{thebibliography}{99}
\small

\bibitem{Aigner}
Aigner,~M.:
\emph{Combinatorial Theory},
Springer, Berlin (1979).

\bibitem{MB}
Baake,~M.: 
Recombination semigroups on measure spaces. 
\textit{Monatsh.\ Math.} \textbf{146}, 267--278 (2005) 
and \textbf{150}, 83--84 (2007)  (Addendum);
\texttt{arXiv:math.CA/0506099}.

\bibitem{reco}
Baake,~M., Baake,~E.:
An exactly solved model for mutation, recombination and selection. 
\textit{Can.\ J.\ Math.} \textbf{55}, 3--41 (2003) 
and \textbf{60}, 264--265 (2008)  (Erratum);
\texttt{arXiv:math.CA/0210422}.

\bibitem{baakeherms}
Baake,~E., Herms,~I.:
Single-crossover dynamics: finite versus infinite populations.
\textit{Bull. Math. Biol.} \textbf{70}, 603--624 (2008).

\bibitem{Bennett}
Bennett,~J.\thinspace H.:
On the theory of random mating.
\textit{Ann.\ Human\ Genetics} \textbf{18}, 311--317 (1954).
 
\bibitem{Christiansen}
Christiansen,~F.\thinspace B.:
\emph{Population Genetics of Multiple Loci}.
Wiley, Chichester (1999).
 
\bibitem{Cohn}
Cohn, D.\thinspace L.:
\emph{Measure Theory}.
Birkh\"{a}user, Boston (1980).

\bibitem{Reinhard}
B\"{u}rger,~R.:
\emph{The Mathematical Theory of Selection, Recombination and Mutation}.
Wiley, Chichester (2000).

\bibitem{Kevin1}
Dawson,~K.\thinspace J.: 
The decay of linkage disequilibria under random union of gametes:\ How to 
calculate Bennett's principal components.
\textit{Theor.\ Popul.\ Biol.} \textbf{58}, 1--20 (2000).

\bibitem{Kevin2}
Dawson,~K.\thinspace J.:
The evolution of a population under recombination: How to linearise the dynamics.
\textit{Lin.\ Alg.\ Appl.} \textbf{348}, 115--137 (2002).

\bibitem{Geiringer}
Geiringer,~H.:
On the probability theory of linkage in Mendelian heredity.
\textit{Ann.\ Math.\ Stat.} \textbf{15}, 25--57 (1944).

\bibitem{Hartl}
Hartl,~D.\thinspace L., Clark,~A.\thinspace G.:
\emph{Principles of Population Genetics}.
 3rd ed, Sinauer, Sunderland, MA (1997).

\bibitem {Jennings}
Jennings,~H.\thinspace S.:
The numerical results of diverse systems of breeding, with respect to two pairs of characters, linked or independent, with special relation to the effects of linkage.
\textit{Genetics} \textbf{2}, 97--154 (1917). 

\bibitem{Lyubich}
Lyubich,~Y.\thinspace I.:
\emph{Mathematical Structures in Population Genetics}.
Springer, Berlin (1992).

\bibitem {HaleRingwood}
McHale,~D., Ringwood,~G.\thinspace A.: 
  Haldane linearisation of baric algebras.
\textit{ J.\ London\ Math.\ Soc.} (2) \textbf{28}, 17--26 (1983).

\bibitem{Popa}
Popa,~E.:
Some remarks on a nonlinear semigroup acting on positive measures.
In: Carja,~O., Vrabie,~I.\thinspace I. (Eds.),
Applied Analysis and Differential Equations,
World Scientific, Singapore, 308--319 (2007).

\bibitem {Robbins}
Robbins,~R.\thinspace B.:
Some applications of mathematics to breeding problems III.
\textit{Genetics} \textbf{3}, 375--389 (1918).

\bibitem{Ute}
von Wangenheim,~U.:
\textit{Diskrete Rekombinationsdynamik}.
Diplomarbeit, Universit\"{a}t Greifswald (2007).

\end{thebibliography}
\end{document}